\documentclass[journal]{IEEEtran}

\usepackage[latin9]{inputenc}
\usepackage{color}
\usepackage{amsmath}
\usepackage{amsthm}
\usepackage{amssymb}
\usepackage{graphicx}
\usepackage[T1]{fontenc}
\usepackage[active]{srcltx}
\usepackage{color}
\usepackage{float}
\usepackage{setspace}
\usepackage{enumitem}
\usepackage[font={footnotesize,it}]{caption}
\usepackage[font={footnotesize,it}]{subcaption}
\usepackage{epsfig}
\usepackage[noadjust]{cite}
\usepackage{verbatim}
\usepackage{balance}

\makeatletter

\newtheorem{theorem}{Theorem}
\newtheorem{proposition}[theorem]{Proposition}
\newtheorem{lemma}[theorem]{Lemma}

\newtheorem{remark}[theorem]{Remark}

\newcounter{algo}
\newenvironment{algo}[2]{\refstepcounter{algo}\label{#2}   \begin{center}
\begin{minipage}{0.49\textwidth}   \hrule\smallskip
\textbf{Algorithm \thealgo: #1}
\par\smallskip\hrule\smallskip\ignorespaces}{\par\smallskip\hrule
\end{minipage}
\end{center}}

\let\oldproofname=\proofname
\renewcommand{\proofname}{\rm\bf{\oldproofname}}

\renewcommand{\Re}{\mathbb{R}}

\newcommand{\z}{\mathbf{z}}
\newcommand{\y}{\mathbf{y}}
\newcommand{\x}{\mathbf{x}}
\newcommand{\trt}{^{\scriptscriptstyle T}}

\newcommand{\A}{\mathbf{A}}
\newcommand{\bv}{\mathbf{b}}
\newcommand{\Q}{\mathbf{Q}}
\begin{document}

\title{\vspace{-0.5cm}Parallel Selective Algorithms  for \\ Nonconvex Big Data Optimization}

\author{Francisco Facchinei,$^\ast$ Gesualdo Scutari,$^\ast$ and  Simone Sagratella\vspace{-0.9cm}
\thanks{$^\ast$ The first two authors contributed equally to the paper.} \thanks{F. Facchinei and S. Sagratella are with the  Dept. of Computer, Control,
and Management Engineering, at Univ.  of Rome La Sapienza, Rome, Italy. Emails: \texttt{<facchinei,sagratella>@dis.uniroma1.it}.}
\thanks{G. Scutari  is with the Dept. of Electrical Engineering, at  the State Univ. of New
York at Buffalo, Buffalo, USA. Email: \texttt{gesualdo@buffalo.edu}}
\thanks{Part of this work has been presented at the
2014 IEEE International Conference on Acoustics, Speech, and Signal Processing (ICASSP 2014), Florence, Italy,
May 4-9 2014, \cite{FaccSagScu_ICASSP14}.}
}

\maketitle
\begin{abstract}
We propose  a  decomposition framework for the parallel optimization
of the sum of a differentiable (possibly nonconvex) function and a (block) separable
nonsmooth, convex one. The latter term is usually employed to enforce structure in the solution, typically sparsity.
Our framework is very flexible and includes both fully parallel Jacobi schemes
and Gauss-Seidel (i.e., sequential) ones, as well as virtually all possibilities ``in between'' with only a subset of  variables  updated at each iteration. Our theoretical convergence results improve on existing ones, and
  numerical results on LASSO,  logistic regression, and some nonconvex quadratic  problems show  that the new method consistently outperforms  existing
algorithms.
\end{abstract}

\vspace{-0.2cm}

\section{Introduction}

\label{sec:Introduction}

The minimization of the sum of a smooth  function, $F$, and of a nonsmooth
(block separable) convex one, $G$,
\begin{equation}\label{eq:problem 1}
\min_{\x\in X} V(\x) \triangleq F(\x) + G(\x),
\end{equation}
is an ubiquitous problem that arises in many fields of
engineering, so diverse as compressed sensing, basis pursuit denoising,
sensor networks, neuroelectromagnetic imaging, machine learning,
data mining, sparse logistic regression,  genomics,  metereology, tensor factorization and completion,
geophysics, and radio astronomy.
Usually the nonsmooth term is used to promote sparsity of the optimal solution,
which often corresponds to a parsimonious representation of some phenomenon at hand.
Many of the aforementioned applications
can give rise to extremely large problems so that standard optimization techniques are hardly applicable.
And indeed, recent years have witnessed a flurry of research activity  aimed at developing
solution methods that are simple (for example based solely on matrix/vector multiplications) but yet capable to
converge to a good approximate solution in reasonable time. It is hard  here to even summarize   the huge amount
of work  done in this field; we refer the reader to the recent works
 \cite{tibshirani1996regression,qin2010efficient,rakotomamonjy2011surveying,yuan2010comparison,byrd2013inexact,fountoulakis2013second,necoara2013efficient,nesterov2012gradient,
 tseng2009coordinate,beck_teboulle_jis2009,wright2009sparse,peng2013parallel,razaviyayn2013unified} and books
 \cite{buhlmann2011statistics,Sra-Nowozin-Wright_book11,bach2011optimization}
 as entry points to the literature.

  However, with big data problems it is clearly necessary to design
\emph{parallel} methods able to  exploit the computational  power of multi-core processors in order to solve many
interesting problems. Furthermore, even if parallel methods are used, it might be useful to reduce the number of (block) variables that are updated at each iteration, so as to alleviate the burden of dimensionality, and also because it has been observed experimentally (even if in restricted settings and under some very strong convergence assumptions,
see \cite{li2009coordinate,NIPS2011_4425,ScherrerTewariHalappanavarHaglin2012parallel_random,peng2013parallel})
that this might be beneficial.
While sequential solutions methods for Problem \eqref{eq:problem 1} have been widely investigated (especially when $F$ is
convex),
the analysis of  parallel algorithms suitable to large-scale implementations  lags behind.
Gradient-type methods can of course be easily parallelized, but, in spite of their good theoretical convergence bounds.  they
suffer from
practical drawbacks.
Fist of all,    they are notoriously  slow in practice.  Accelerated (proximal) versions have been proposed in the literature to
alleviate this issue, but  they require the knowledge of some function parameters (e.g.,  the Lipschitz constant of $\nabla F$,
and   the strong convexity constant of $F$, when $F$ is assumed strongly convex), which  is not generally available, unless
$F$ and $G$ have a very special structure (e.g., quadratic functions); (over)estimates of such parameters affect negatively
the convergence speed. Moreover, all (proximal, accelerated) gradient-based schemes use only the first order information of
$F$; recently we showed in   \cite{scutari_facchinei_et_al_tsp13} that  exploiting the structure of $F$ by replacing its
linearization with a ``better'' approximant can  enhance practical convergence speed.
However, beyond that, and looking at recent approaches,
we are  aware of  only few  (recent) papers
  dealing with parallel solution methods   \cite{necoara2013efficient,nesterov2012gradient,
 tseng2009coordinate,beck_teboulle_jis2009,wright2009sparse,peng2013parallel} and
   \cite{nesterov2012efficiency,bradley2011parallel,necoara_nesterov_glineur2011,
   necoara_clipici2014,LuXiao2013randomCDM,
 ScherrerTewariHalappanavarHaglin2012parallel_random,
 richtarik2012parallel}.
The former group of works investigate {\em deterministic} algorithms, while the latter {\em random} ones.
One advantage of the analyses in these works is that they provide a global rate of convergence. However,
i) they are essentially  still (regularized) gradient-based methods;  ii) as proximal-gradient algorithms,  they require  good and global knowledge  of some $F$ and $G$ parameters; and iii) except for \cite{nesterov2012gradient,
 tseng2009coordinate,wright2009sparse,LuXiao2013randomCDM}, they are proved to converge only when $F$ is  convex.
 Indeed, to date  there are no methods that are  parallel and  random and  that can be applied to \emph{nonconvex} problems. For this reason, and also because of their markedly  different flavor (for example deterministic convergence vs. convergence in mean or in probability), we do not discuss further random algorithms in this paper.
We refer instead to Section \ref{sec:related_works} for a more detailed discussion on   deterministic, parallel, and sequential solution methods proposed in the literature for instances (mainly convex) of  (\ref{eq:problem 1}).

In this paper, we focus on   \emph{nonconvex} problems in the form (1), and proposed  a new broad,  \emph{deterministic}   algorithmic framework for the computation of their stationary solutions, which hinges on ideas first introduced in \cite{scutari_facchinei_et_al_tsp13}. The essential, rather natural  idea underlying our  approach is 
to decompose
\eqref{eq:problem 1} into a sequence of (simpler)  subproblems whereby the function $F$ is replaced by suitable convex
approximations; the subproblems can be solved in a \emph{parallel} and \emph{distributed} fashion and it is not necessary to update all variables at each iteration.   
Key   features of the proposed algorithmic framework are:
i) it is parallel, with a degree of parallelism that can be chosen by the user and that can go from a complete parallelism  (every variable is updated in parallel to all the others) to the sequential
(only one variable is updated at each iteration), covering virtually all the possibilities in ``between'';
ii) it permits the  update in a \emph{deterministic fashion } of only some (block)  variables at each iteration (a feature  that turns out to be very important
numerically); iii) it easily leads to distributed implementations;
iv) no knowledge of $F$ and $G$ parameters (e.g., the Lipschitz constant of $\nabla F$) is required;
v) it can tackle a nonconvex $F$;
vi)   it is very flexible in the choice of the approximation of $F$, which   need not be necessarily  its first or second order approximation (like in proximal-gradient schemes); of course it  includes, among others,  updates based on gradient- or Newton-type approximations;
vii) it easily allows for inexact solution of the subproblems  (in some large-scale problems the cost of computing the exact  solution of all the subproblems can be prohibitive); and
viii) it appears to be numerically efficient.
While  features i)-viii), taken singularly, are certainly not new in the optimization community, we are not aware of any algorithm that possesses them all, the more so if one limits attention  to methods that can handle nonconvex objective functions, which is in fact the main focus on this paper.
Furthermore, numerical results show that our scheme compares favorably to existing ones.

The proposed  framework encompasses a gamut of novel algorithms, offering great flexibility  to control iteration complexity, communication overhead, and convergence speed, while \emph{converging under the  same conditions}; these desirable features  make our schemes applicable to several different problems and scenarios. 
Among the variety of new proposed updating rules for the (block)  variables, it is worth mentioning a combination of Jacobi  and Gauss-Seidel updates, which seems particularly
valuable in the optimization of highly nonlinear objective function; to the best of our knowledge, this is the first time that such a scheme is proposed and analyzed.

A further contribution of the paper is an extensive  implementation effort over a parallel architecture (the General Computer Cluster of the Center for Computational Research at the SUNY at Buffalo), which includes  our schemes and the  most competitive ones in the literature. 
Numerical results on LASSO, Logistic Regression, and some nonconvex problems show that our algorithms consistently outperform state-of-the-art schemes.

 The paper is organized as follows. Section \ref{sec:Pro} formally introduces the optimization  problem  along with the main assumptions  under which it is studied.
 Section \ref{sec:Main Results} describes our novel  general algorithmic framework along with its convergence properties. 
In  Section \ref{sec:Exe} we discuss several instances of the general scheme introduced in Section   \ref{sec:Main Results}.  Section \ref{sec:related_works} contains a detailed comparison of our schemes with state-of-the-art algorithms on similar problems.  Numerical results are presented in Section  \ref{sec:Num}. 
  Finally, Section \ref{sec:conclusions} draws some conclusions. All proofs of our results are given in the Appendix.\vspace{-0.1cm}

\section{Problem Definition}\label{sec:Pro}\vspace{-0.1cm}

We consider Problem \eqref{eq:problem 1},
where the feasible set $X=X_1\times \cdots \times  X_N$ is a Cartesian product of lower dimensional convex sets $X_i\subseteq \Re^{n_i}$, and $\mathbf x\in \Re^n$ is partitioned accordingly: $\mathbf x = (\mathbf  x_1, \ldots, \mathbf x_N)$, with each $\mathbf  x_i \in \Re^{n_i}$; $F$ is  smooth (and not necessarily convex)
and $G$  is convex and possibly nondifferentiable, with  $ G(\x) = {\scriptstyle \sum_{i=1}^N} g_i(\x_i)$.
This formulation is very general and includes problems of great interest.
Below we list some  instances of Problem \eqref{eq:problem 1}.

\noindent $\bullet$
$G(\x)=0$; the problem reduces to the minimization of a smooth, possibly nonconvex problem with convex constraints.

\noindent $\bullet$
$ F(\x) = \|\A\x -\bv\|^2 $ and $G(\x)= c \|\x\|_1$, $X=\Re^n$, with $\A\in\Re^{m \times n}$, $\bv \in \Re^m$,  and $c\in \Re_{++}$    given constants; this is the renowned and much studied LASSO problem \cite{tibshirani1996regression}.

\noindent $\bullet$
 $ F(\x) = \|\A\x -\bv\|^2 $ and $G(\x)= c \sum_{i=1}^N\|\x_i\|_2$, $X=\Re^n$, with $\A\in\Re^{m \times n}$, $\bv \in \Re^m$,  and $c\in \Re_{++}$
given constants; this is the group LASSO problem \cite{yuan2006model}.

\noindent $\bullet$
$ F(\x) = \sum_{j=1}^m \log(1 + e^{-a_j \y_j^T\x }) $ and $G(\x)= c \|\x\|_1$ (or $G(\x)= c \sum_{i=1}^N\|\x_i\|_2$), with $\y_i\in \Re^n$, $a_i\in \Re$, and $c\in \Re_{++}$  given
constants; this is the sparse logistic regression problem \cite{shevade2003simple,meier2008group}.

\noindent $\bullet$
$F(\x) = \sum_{j=1}^m\max\{0, 1- a_i \y_i^T\x\}^2$ and $G(\x)= c \|\x\|_1$, with $a_i\in \{-1,1\}$,  $\y_i\in \Re^n$, and  $c\in \Re_{++}$ given; this is
the $\ell_1$-regularized
$\ell_2$-loss Support Vector Machine problem  \cite{yuan2010comparison}.

\noindent $\bullet$ $F(\mathbf{X}_1,\mathbf{X}_2)=\|\mathbf{Y}-\mathbf{X}_1\mathbf{X}_2\|_F^2$ and $G(\mathbf{X}_2)=c\|\mathbf{X}_2\|_1$, $X= \{(\mathbf{X}_1,\mathbf{X}_2)\in \Re^{n\times m}\times \Re^{m\times N}\,:\, \|\mathbf{X}_{1}\mathbf{e}_i\|^2\leq \alpha_i,\forall i=1,\ldots,m \}$, where     $\mathbf{X}_1$ and $\mathbf{X}_2$   are the (matrix) optimization variables,   $\mathbf{Y}\in \Re^{n\times N}$, $c>0$, and $(\alpha_i)_{i=1}^m>\mathbf{0}$ are given constants, $\mathbf{e}_i$ is the $m$-dimensional  vector with a 1 in the $i$-th coordinate and $0$'s elsewhere, and $\|\mathbf{X}\|_F$ and $\|\mathbf{X}\|_1$ denote the Frobenius norm and the  L1 matrix norm of  $\mathbf{X}$, respectively; this is an example of the dictionary learning problem for sparse representation, see, e.g., \cite{LeeBattleRainaNg2007dictionary_learning}. Note that $F(\mathbf{X}_1,\mathbf{X}_2)$ is not jointly convex in $(\mathbf{X}_1,\mathbf{X}_2)$.\vspace{-0.01cm}

\noindent $\bullet$  Other problems that can be cast in the form \eqref{eq:problem 1} include the Nuclear Norm Minimization problem,  the Robust Principal Component Analysis problem, the Sparse Inverse Covariance Selection problem, the Nonnegative Matrix (or Tensor) Factorization problem, see e.g., \cite{goldfarb2012fast} and references therein.
 \smallskip

\noindent \textbf{Assumptions.} Given \eqref{eq:problem 1}, we make the following  blanket assumptions:

\begin{description}[topsep=-2.0pt,itemsep=-2.0pt]
\item[\rm  (A1)]  Each $X_i$ is nonempty, closed, and convex;\smallskip
\item[\rm  (A2)] $F$ is $C^1$ on an open set containing $X$;\smallskip
\item[\rm  (A3)]  $\nabla F$ is   Lipschitz continuous
on $X$ with constant $L_{F}$;\smallskip
\item[\rm  (A4)] $G(\x) = \sum_{i=i}^N g_i(\x_i)$, with all $g_i$ continuous and convex on $X_i$;\smallskip
\item[\rm  (A5)] $V$ is coercive.
\end{description}
\medskip

\noindent
Note that the above assumptions are standard and are satisfied by most of the problems of practical interest. For instance,  A3 holds automatically if $X$ is bounded;  the block-separability condition   A4 is   a common assumption  in the literature of parallel methods for the class of problems  \eqref{eq:problem 1} (it is in fact  instrumental  to deal with the nonsmoothness of
$G$ in a parallel environment). Interestingly A4 is satisfied by all standard $G$ usually encountered in applications, including  $G(x) = \| x\|_1$ and
$ G(x) = \sum_{i=1}^N\|\x_i\|_2$, which are among the most commonly used functions. 
Assumption A5 is needed to guarantee that the sequence generated by our
method is bounded; we could dispense with it at the price of a more complex analysis and cumbersome statement of convergence results. 


\section{Main Results\label{sec:Main Results}}

We begin introducing an informal description of our  algorithmic framework   along with  a list of key features that we would like our schemes enjoy; this will shed light on the core idea of the proposed decomposition technique. 

We want  to develop  {\em parallel} solution methods for Problem \eqref{eq:problem 1} whereby operations can be carried out on some or  (possibly) all   (block) variables $\x_i$ at
the \emph{same} time. The most natural   parallel (Jacobi-type) method one can think of is  updating  \emph{all} blocks simultaneously: given $\mathbf{x}^k$, each (block) variable $\mathbf{x}_i$  is updated   by solving the following subproblem
\begin{equation}\label{eq:plain_Jac}
\mathbf{x}_i^{k+1}\in  \underset{\mathbf{x}_i\in X_i}{\text{argmin}}\,\, \left\{ F(\mathbf{x}_i, \mathbf{x}_{-i}^k) + g_i(\mathbf{x}_i)\right\},
\end{equation}
where $\x_{-i}$ denotes the vector obtained from $\mathbf{x}$ by deleting the block $\x_i$.
Unfortunately this method converges only under very restrictive conditions  \cite{Bertsekas_Book-Parallel-Comp} that are seldom verified in practice. To cope with this issue  the proposed approach   introduces some ``memory" in the iterate:  the new point is  a convex combination of $\mathbf{x}^k$ and the solutions of (\ref{eq:plain_Jac}).
Building on this iterate, we would like  our framework to enjoy  many additional features, as described  next. 

\noindent \textbf{Approximating $F$:} Solving each subproblem as in (\ref{eq:plain_Jac}) may be too costly or difficult in
some
situations. One may then
prefer to approximate this  problem, in some suitable sense,  in order to facilitate the task of computing the new iteration.
To this end,
we assume that for all  $i\in {\cal N} \triangleq \{1, \dots, N\}$ we
can define  a function  $P_{i} (\mathbf{z};\mathbf{w}) :  X_i \times  X \to \Re$, the candidate approximant of $F$,   having the following properties (we denote by $\nabla P_{i}$
the partial gradient of $P_i$ with respect to the first argument   $\mathbf{z}$):\smallskip
\begin{description}
\item[\rm  (P1)]
 $P_{i} (\mathbf{\bullet}; \mathbf{w})$ is  convex and continuously differentiable
on $ X_i$ for all $\mathbf{w}\in X$;
\item[\rm  (P2)]  $\nabla P_{i} (\mathbf{x}_i;\mathbf{x}) = \nabla_{{\mathbf{x}}_i} F(\mathbf{x})$ for all $\mathbf{x} \in  X$;
\item[\rm  (P3)]  $\nabla P_{i} (\mathbf{z};\mathbf{\bullet})$ is Lipschitz continuous
on $ X$ 
 for all $\mathbf{z} \in  X_i$.
\end{description}
\smallskip

Such a  function  $P_i$ should be regarded as a (simple) convex approximation of $F$ at the point $\mathbf{x}$ with respect to the block of variables $\mathbf{x}_i$ that
preserves the first order properties of $F$ with respect to  $\mathbf{x}_i$.

Based on this approximation we can define at any point $\x^k\in X$  a {\em regularized} approximation  $\widetilde{h}_i (\mathbf{x}_{i};\mathbf{x}^k)$ of $V$ with respect
to $\mathbf{x}_i$ wherein $F$ is replaced  by $P_i$ while  the nondifferentiable term is preserved,  and a quadratic
proximal term is added to make the overall approximation strongly convex. More formally, we have\vspace{-0.1cm}
\begin{equation}\label{q_approx}
\hspace{-0.4cm}\begin{array}{ll}
& \widetilde{h}_i (\mathbf{x}_{i}; \mathbf{x}^k) \triangleq \underbrace{  P_{i}(\mathbf{x}_{i}; \mathbf{x}^k)
 +\dfrac{{\tau}_{i}}{2}  \left(\mathbf{x}_{i}- \mathbf{x}_{i}^k\right)^{T} \Q_{i}( \mathbf{x}^k) \left(\mathbf{x}_{i}- \mathbf{x}^k_{i}\right)}_{\triangleq h_i(\mathbf{x}_{i}; \mathbf{x}^k) }
 \\ &\qquad\qquad\,\,\quad+ g_{i}(\mathbf{x}_{i}),
\end{array} 
\end{equation}
where $\Q_{i}(\mathbf{x}^k)$ is an $n_{i}\times n_{i}$
 positive definite matrix (possibly dependent on $\mathbf{x}^k)$. We always assume that the functions $h_i(\bullet, \mathbf{x}_{i}^k)$ are uniformly strongly convex. \smallskip

 \begin{description}
\item[\rm  (A6)] All   $h_{i}(\bullet;\mathbf{x}^k)$ are uniformly strongly convex on $X_i$ with a common positive definiteness constant $q > 0$; furthermore,  $\Q_{i}(\bullet)$ is Lipschitz continuous on $X$.\smallskip
\end{description}
Note that an easy and standard way to satisfy A6 is to take, for any $i$ and for any $k$, $\tau_i = q > 0$ and $ \Q_{i}( \mathbf{x}^k) = \mathbf{I}$. However,  if $P_i(\bullet; \mathbf{x}^k)$ is already uniformly strongly convex, one can avoid the proximal term  and set  $\tau_i =0$  while satisfying A6.

Associated with each $i$ and point $ \mathbf{x}^k \in X$ we can define the following optimal block solution map:\vspace{-0.1cm}
\begin{equation}\label{eq:decoupled_problem_i}
\widehat{\mathbf{x}}_{i}(\mathbf{x}^k,\tau_{i})\triangleq\underset{\mathbf{x}_{i}\in X_{i}}{\mbox{argmin}\,}\tilde{h}_{i}(\mathbf{x}_{i};\mathbf{x}^{k}).\vspace{-0.1cm}
\end{equation}
 Note that $\widehat{\mathbf{x}}_{i}(\mathbf{x}^{k},\tau_{i})$
is always well-defined, since the optimization problem in (\ref{eq:decoupled_problem_i})
is strongly convex. Given (\ref{eq:decoupled_problem_i}),
we can then introduce, for each $\mathbf{y} \in X$,  the solution map
\[
\widehat{\x}(\mathbf{y},\boldsymbol{{\tau}})\triangleq\large[\widehat{\x}_{1}(\y,\tau_{1}), \widehat{\x}_{2}(\y,\tau_{2}),
\ldots, \widehat{\x}_{N}(\y,\tau_{N})\large]^T.
\]
The proposed algorithm (that we formally describe later on) is based on the computation of  (an approximation of)  $\widehat{\x}(\mathbf{x}^{k},\boldsymbol{\tau})$. Therefore the  functions $P_i$ should
lead to as easily computable functions $\widehat{\x}$ as possible. An appropriate choice depends on the problem at hand and   on computational
requirements. We discuss alternative  possible choices for the  approximations $P_i$ in Section \ref{sec:Exe}. 

\noindent \textbf{Inexact solutions:}  In many situations (especially in the case of large-scale problems), it can be useful to
further reduce the computational effort needed to solve the subproblems in \eqref{eq:decoupled_problem_i} by allowing \emph{inexact} computations $\mathbf{z}^k$ of $\widehat{\mathbf{x}}_{i}(\mathbf{x}^k,\tau_{i})$, i.e.,   $\|\mathbf{z}_{i}^{k}-\widehat{\mathbf{x}}_i\left(\mathbf{x}^{k},\boldsymbol{{\tau}}\right)\|\leq\varepsilon_{i}^{k}$, where $\varepsilon_{i}^{k}$ measures the   accuracy in computing the solution. 

\noindent \textbf{Updating  only some blocks:} Another important feature we want for our algorithm  is the  capability of updating at each iteration only \emph{some} of the (block) variables, a feature that has been observed to be very effective numerically.
In fact, our schemes are guaranteed to converge  under the update of only a \emph{subset} of the
 variables at each iteration;  the only condition is that such a  subset contains at least one (block) component
which is within a factor $\rho \in (0,1]$   ``far away'' from the optimality, in the sense explained next.
Since  $\x^k_i$ is an optimal solution of \eqref{eq:decoupled_problem_i} if and only if
$\widehat{\x}_{i}(\x^k,\tau_{i})=\x^k_i$, a natural  distance of  $\x^k_i$ from the optimality is  $d_i^{\,k}\triangleq \|\widehat{\x}_{i}(\x^k,
\tau_{i})-\x^k_i\|$; one could then  select the  blocks $\x_i$'s to update based on  such an  optimality measure (e.g., opting for blocks exhibiting larger $d_i^{\,k}$'s). However, this choice requires the  computation of all the solutions $\widehat{\x}_{i}(\x^k,
\tau_{i})$, for   $i=1,\ldots,n$, which in some applications (e.g., huge-scale problems) might be  computationally  too expensive.
Building on the same idea, we can introduce alternative less expensive metrics by replacing the distance   $\|\widehat{\x}_{i}(\x^k,
\tau_{i})-\x^k_i\|$ with a computationally cheaper {\em error bound},  
i.e., a function $E_i(\x)$ such that\vspace{-0.1cm}
\begin{equation}\label{eq:error bound}
\underbar s_i\|\widehat{\x}_{i}(\x^k,\tau_{i})-\x^k_i\| \le E_i(\x^k) \leq \bar s_i\|\widehat{\x}_{i}(\x^k,\tau_{i})-\x^k_i\|,\vspace{-0.1cm}
\end{equation}
for some $0< \underbar s_i \le \bar s_i$. Of course one can always set  $ E_i(\x^k) = \|\widehat{\x}_{i}(\x^k,\tau_{i})-\x^k_i\| $, but other choices are also possible; we discuss
this point further in Section \ref{sec:Exe}.\smallskip

\noindent \textbf{Algorithmic framework:}  We are now ready to formally introduce our algorithm, Algorithm 1,  that includes all the features discussed above;  convergence to stationary solutions\footnote{We recall that a stationary solution  $\x^*$ of \eqref{eq:problem 1} is a   points  for which  a subgradient $\xi \in \partial G(\x^*)$ exists such that
$(\nabla F(\x^*) $ $+ \xi)^T(\y-\x^*) \geq 0$ for all $\y\in X$. Of course,  if $F$ is convex, stationary points coincide with  global minimizers.} of  \eqref{eq:problem 1} is stated in  Theorem \ref{Theorem_convergence_inexact_Jacobi}. 

\vspace{-0.2cm}

\begin{algo}{\textbf{Inexact Flexible Parallel Algorithm (FLEXA)}} S$\textbf{Data}:$ $\{\varepsilon_{i}^{k}\}$
for $i\in\mathcal{N}$, $\boldsymbol{{\tau}}\geq\mathbf{0}$, $\{\gamma^{k}\}>0$,
$\mathbf{x}^{0}\in X$, $\rho \in (0,1]$.

\hspace{1cm} Set $k=0$.

\texttt{$\mbox{(S.1)}:$}$\,\,$If $\mathbf{x}^{k}$ satisfies a termination
criterion: STOP;

\texttt{$\mbox{(S.2)}:$}  Set $M^k \triangleq  \max_i \{E_i(\x^k)\}$.

\hspace{1.24cm} Choose a set $S^k$  that  contains at least one index $i$

\hspace{1.24cm} for which
$E_i(\x^k) \geq \rho M^k.$

\texttt{$\mbox{(S.3)}:$} For all $i\in S^k$, solve (\ref{eq:decoupled_problem_i})
with accuracy $\varepsilon_{i}^{k}:$

\hspace{1.26cm} Find $\mathbf{z}_{i}^{k}\in X_i$ s.t. $\|\mathbf{z}_{i}^{k}-\widehat{\mathbf{x}}_i\left(\mathbf{x}^{k},\boldsymbol{{\tau}}\right)\|\leq\varepsilon_{i}^{k}$;

\hspace{1.24cm} Set $\widehat{\z}^k_i =  \mathbf{\z}_{i}^k$ for $i\in S^k$ and
$\widehat{\z}^k_i = \x^k_i$ for $i\not \in S^k$

\texttt{$\mbox{(S.4)}:$} Set $\mathbf{x}^{k+1}\triangleq\mathbf{x}^k+\gamma^{k}\,(\widehat{\mathbf{z}}^{k}-\mathbf{x}^{k})$;

\texttt{$\mbox{(S.5)}:$} $k\leftarrow k+1$, and go to \texttt{$\mbox{(S.1)}.$}
\label{alg:general}
 \end{algo}\smallskip

\begin{theorem} \label{Theorem_convergence_inexact_Jacobi}Let
$\{\mathbf{x}^{k}\}$ be the sequence generated by
Algorithm \ref{alg:general}, under A1-A6.
 Suppose  that $\{\gamma^{k}\}$
and $\{\varepsilon_{i}^{k}\}$ satisfy the following conditions: i)
$\gamma^{k}\in(0,1]$;  ii) $\sum_{k}\gamma^{k}=+\infty$;
iii) $\sum_{k}\left(\gamma^{k}\right)^{2}<+\infty$;
and iv)  $\varepsilon_{i}^{k}\leq \gamma^k  \alpha_1\min\{\alpha_2, 1/\|\nabla_{\mathbf{x}_i} F(\mathbf{x}^k)\| \}$
for all $i\in {\cal N}$ and  some $\alpha_1>0$ and $\alpha_2>0$.
Additionally, if inexact solutions are used in Step 3, i.e.,  $\varepsilon_{i}^{k}>0$ for some $i$ and infinite $k$, then
assume also that $G$ is globally Lipschitz on $X$.
Then, either Algorithm \ref{alg:general} converges in a finite number of iterations to a stationary solution
of \eqref{eq:problem 1} or every limit point of
  $\{\mathbf{x}^{k}\}$ (at least
one such points exists) is a stationary solution of \eqref{eq:problem 1}.\vspace{-0.1cm}
\end{theorem}
\begin{proof} See Appendix \ref{proof_Th1}.\end{proof}\smallskip

 The proposed algorithm is extremely flexible. We can always choose $S^k ={\cal N}$
resulting in the simultaneous  update of all the (block) variables (full Jacobi scheme); or, at the other extreme, one can
update a single (block) variable per time, thus obtaining a Gauss-Southwell kind of method. More classical cyclic
Gauss-Seidel methods can also be   derived and are discussed  in the next subsection.
One can also compute inexact solutions (Step 3) while preserving convergence, provided that  the error term $\varepsilon_{i}
^{k}$  and  the step-size $\gamma^{k}$'s are  chosen according
to Theorem 1; some practical  choices for these parameters  are discussed in  Section \ref{sec:Exe}.  We emphasize that the Lipschitzianity of $G$ is required only if $\widehat{\mathbf{x}}(\mathbf{x}^k, \tau)$ is
not computed exactly for \emph{infinite} iterations. At any rate this Lipschitz conditions is automatically satisfied if $G$ is a norm
(and therefore in LASSO and  group LASSO problems for example) or if
$X$ is bounded.

As a final remark,  note  that versions of Algorithm \ref{alg:general} where all  (or most of) the variables are updated at each iteration are particularly amenable to implementation in \emph{distributed} environments (e.g., multi-user communications systems, ad-hoc networks, etc.). In fact, in this case, not only   the calculation
of the inexact solutions  $z^k_i$ can be carried out in parallel, but the information that ``the $i$-th subproblem'' has to exchange with the ``other subproblem'' in order to compute the next iteration is very limited. A full appreciation  of the potentialities
of our approach in distributed settings depends however on the specific application under  consideration and is beyond the scope of this paper. We refer the reader to \cite{scutari_facchinei_et_al_tsp13} for some examples, even if in less general settings.\vspace{-0.1cm}

\subsection{Gauss-Jacobi algorithms}\label{subsec:GS}\vspace{-0.1cm}
Algorithm \ref{alg:general} and its convergence theory  cover fully parallel Jacobi as well as Gauss-Southwell-type methods, and many of their variants.
In this section we show that Algorithm 1 can also  incorporate \emph{hybrid parallel-sequential} (Jacobi$-$Gauss-Seidel) schemes wherein block of variables are updated \emph{simultaneously} by \emph{sequentially} computing entries per block. This procedure seems particularly well suited to parallel optimization on   multi-core/processor architectures. 

Suppose that we have $P$ processors that can be used in parallel and we want to exploit them to
solve Problem \eqref{eq:problem 1} ($P$ will denote both the number of processors and the set $\{1,2,\ldots, P\}$).
We ``assign'' to each processor $p$ the variables $I_p$; therefore $I_1, \ldots, I_P$
is a partition of $I$. We denote by $\x_p \triangleq (\x_{pi})_{i\in I_p}$ the vector of (block)  variables $\x_{pi}$ assigned to processor $p$, with $i \in I_p$;  
and  $\x_{-p}$ is the vector of remaining variables, i.e., the vector of those   assigned to
all processors except the $p$-th one.
Finally, 
given $i\in I_p$,
we partition $\x_p$ as $\x_p = (\x_{pi<}, \x_{pi\geq})$, where $\x_{pi<}$ is the vector containing
all variables in $I_p$ that come before $i$ (in the order assumed in $I_p$), while   $\x_{pi\geq}$
are the remaining variables. Thus we will write, with a slight  abuse of notation
$\x = (\x_{pi<}, \x_{pi\geq},\x_{-p})$.

Once the optimization variables have been assigned to the processors, one could in principle apply the
inexact Jacobi  Algorithm \ref{alg:general}. In
this scheme each processor $p$ would compute {\em sequentially}, at each iteration $k$ and for every (block) variable
$\x_{pi}$, a suitable
$\mathbf{z}^k_{pi}$ by keeping all variables but $\x_{pi}$ fixed to $(\x^k_{pj})_{i\neq j\in I_p}$ and $\x^k_{-p}$. But since we are solving the
problems for each group of
variables assigned to a processor sequentially, this seems a  waste of resources; it is instead  much more
efficient to use, within each
processor, a Gauss-Seidel scheme, whereby the \emph{current} calculated  iterates  are used in all subsequent
calculations.
Our Gauss-Jacobi method formally described in Algorithm 2  implements  exactly this idea; its convergence properties are given in Theorem 2.\vspace{-0.3cm}
\begin{algo}{\textbf{Inexact Gauss-Jacobi Algorithm}} S$\textbf{Data}:$
$\{\varepsilon_{pi}^{k}\}$ for $p\in P$ and  $i\in{I}_p$, $\boldsymbol{{\tau}}\geq\mathbf{0}$,
$\{\gamma^{k}\}>0$, $\mathbf{x}^{0}\in\mathcal{K}$.

\hspace{1.24cm} Set $k=0$.

\texttt{$\mbox{(S.1)}:$}$\,\,$If $\mathbf{x}^{k}$ satisfies a termination criterion: STOP;

\texttt{$\mbox{(S.2)}:$} For all $p\in P$ do (in parallel),

\hspace{1.5cm} For all $i\in I_p$ do (sequentially)

\hspace{2.0cm}a) Find $\mathbf{z}^k_{pi} \in \mathcal X_i$ s.t.

\hspace{2.5cm}
$\|\mathbf{z}_{pi}^k-\widehat{{\x}}_{pi}\left((\mathbf{x}_{pi<}^{k+1},\mathbf{x}_{pi\geq}^{k}, \mathbf{x}_{-p}^k),\boldsymbol{{\tau}}\right)\|\leq\varepsilon_{pi}^{k}$;\smallskip

\hspace{2.0cm}b) Set ${\x}_{pi}^{k+1}\triangleq {\x}_{pi}^{k}+{\gamma}^{\, k}\,\left(\mathbf{z}_{pi}^{k}-{\x}_{pi}^{k}\right)$\smallskip

\texttt{$\mbox{(S.3)}:$} $k\leftarrow k+1$, and go to \texttt{$\mbox{(S.1)}.$}
\label{alg:PJGA}\end{algo}

\begin{theorem} \label{Theorem_convergence_Gauss_Jacobi}Let\emph{
$\{\mathbf{x}^{k}\}_{k=1}^{\infty}$} be the sequence generated by
Algorithm \ref{alg:PJGA}, under the setting of Theorem \ref{Theorem_convergence_inexact_Jacobi}, but  with the addition
assumption that $\nabla F$ is bounded on $\mathcal X$.
Then, either Algorithm\emph{ }\ref{alg:PJGA} converges in a finite number of iterations to a stationary solution
of (\ref{eq:problem 1}) or every limit point of
the sequence\emph{ $\{\mathbf{x}^{k}\}_{k=1}^{\infty}$ }(at least
one such points exists) is a stationary solution of (\ref{eq:problem 1}).\vspace{-0.2cm}
\end{theorem}
\begin{proof} See Appendix \ref{proof_Th2}. \end{proof}

Although  the proof of Theorem \ref{Theorem_convergence_Gauss_Jacobi}
is relegated to the appendix, it is interesting to point out that the gist of the proof is to show that  Algorithm \ref{alg:PJGA} is nothing else but an instance of Algorithm \ref{alg:general} with errors.
We remark that Algorithm \ref{alg:PJGA} contains as special case the  classical cyclical Gauss-Seidel scheme
it is sufficient to set  $P=1$
then a single processor   updates all the (scalar) variables sequentially while using the new values of those that have already been updated.\\
\indent By updating  all variables at each iteration,  Algorithm \ref{alg:PJGA} has the advantage that  neither the error bounds $E_i$ nor the exact solutions $\widehat{\mathbf{x}}_{pi}$ need to be computed, in order to decide which variables should be updated. Furthermore it is rather intuitive that the use of the ``latest available information'' should reduce the number of overall iterations  needed to converge with respect to Algorithm \ref{alg:general}.
 However this advantages should be weighted against  the following two  facts: i)  updating all variables at each iteration might not always be the best (or a feasible) choice; and ii) in many practical instances of  Problem (\ref{eq:problem 1}), using the latest information as dictated by Algorithm \ref{alg:PJGA}  may require  extra calculations, e.g. to compute function gradients, and communication overhead  (these aspects are discussed  on specific examples  in  Section \ref{sec:Num}).
It may then be of interest to consider a further scheme, that we might call ``Gauss-Jacobi with Selection'', where we merge the basic ideas of Algorithms
  \ref{alg:general} and \ref{alg:PJGA}. Roughly speaking, at each iteration we proceed as in the Gauss-Jacobi Algorithm  \ref{alg:PJGA}, but we perform the Gauss steps only on a subset ${S}^k_p$ of each $I_p$, where the subset ${S}^k_p$ is defined  according to the rules used in Algorithm
\ref{alg:general}.
To present this combined scheme, we need to extend the notation used in Algorithm  \ref{alg:PJGA}.
Let ${S}^k_p$ be a subset of $I_p$. For notational purposes only, we reorder $\x_p$ so that
first we have all variables in  ${S}^k_p$ and then the remaining variables in $I_p$:
$\x_p= (\x_{{S}^k_p}, \x_{I_p\setminus {S}^k_p})$. Now, similarly to what done before, and given
an index $i \in {S}^k_p$,  we partition $\x_{{S}^k_p}$ as
$\x_{{S}^k_p} = (\x_{{S}^k_pi<}, \x_{{S}^k_pi\geq})$, where $\x_{{S}^k_pi<}$ is the vector containing
all variables in ${S}^k_p$ that come before $i$ (in the order assumed in ${S}^k_p$), while   $\x_{{S}^k_pi\geq}$
are the remaining variables in ${S}^k_p$. Thus we will write, with a slight  abuse of notation
$\x = (\x_{{S}^k_pi<}, \x_{{S}^k_pi\geq},\x_{-p})$. The proposed Gauss-Jacobi with Selection is formally described in Algorithm 3 below.\vspace{-0.5cm}
\begin{algo}{\textbf{Inexact GJ Algorithm with Selection}} S$\textbf{Data}:$
$\{\varepsilon_{pi}^{k}\}$ for $p\in P$ and  $i\in{I}_p$, $\boldsymbol{{\tau}}\geq\mathbf{0}$,
$\{\gamma^{k}\}>0$, $\mathbf{x}^{0}\in\mathcal{K}$,

\hspace{1.24cm}
$\rho \in (0,1]$. Set $k=0$.

\texttt{$\mbox{(S.1)}:$}$\,\,$If $\mathbf{x}^{k}$ satisfies a termination criterion: STOP;

\texttt{$\mbox{(S.2)}:$}  Set $M^k \triangleq  \max_i \{E_i(\x^k)\}$.

\hspace{1.24cm} Choose  sets $S^k_p \subseteq I_p$  so that their union $S^k$ contains

\hspace{1.24cm} at least one index   $i$ for which $E_i(\x^k) \geq \rho M^k.$

\texttt{$\mbox{(S.3)}:$} For all $p\in P$ do (in parallel),

\hspace{1.5cm} For all $i\in S^k_p $ do (sequentially)

\hspace{2.0cm}a) Find $\mathbf{z}^k_{pi} \in \mathcal X_i$ s.t.

\hspace{2.28cm}
$\|\mathbf{z}_{pi}^k-\widehat{{\x}}_{pi}\left(
(\x_{{S}^k_pi<}^{k+1}, \x_{{S}^k_pi\geq}^{k},\x_{-p}^{k}),
\boldsymbol{{\tau}}\right)\|\leq\varepsilon_{pi}^{k}$;\smallskip

\hspace{2.0cm}b) Set ${\x}_{pi}^{k+1}\triangleq {\x}_{pi}^{k}+{\gamma}^{\, k}\,\left(\mathbf{z}_{pi}^{k}-{\x}_{pi}^{k}\right)$\smallskip

\texttt{$\mbox{(S.4)}:$} Set $\x^{k+1}_{pi} = \x^{k}_{pi}$ for all $i \not \in S^k$,

\hspace{1.24cm}
$k\leftarrow k+1$, and go to \texttt{$\mbox{(S.1)}.$}
\label{alg:PJGA  bis}\end{algo} \vspace{-0.1cm}


\begin{theorem} \label{Theorem_convergence_Gauss_Jacobi bis}Let\emph{
$\{\mathbf{x}^{k}\}_{k=1}^{\infty}$} be the sequence generated by
Algorithm \ref{alg:PJGA bis}, under the setting of Theorem \ref{Theorem_convergence_inexact_Jacobi}, but  with the addition
assumption that $\nabla F$ is bounded on $\mathcal X$.
Then, either Algorithm\emph{ }\ref{alg:PJGA bis} converges in a finite number of iterations to a stationary solution
of (\ref{eq:problem 1}) or every limit point of
the sequence\emph{ $\{\mathbf{x}^{k}\}_{k=1}^{\infty}$ }(at least
one such points exists) is a stationary solution of (\ref{eq:problem 1}).\vspace{-0.2cm}
\end{theorem}
\begin{proof} The proof is just a (notationally complicated, but conceptually easy) combination of the proofs of  Theorems \ref{Theorem_convergence_inexact_Jacobi} and \ref{Theorem_convergence_Gauss_Jacobi} and is omitted for lack of space. \end{proof}\smallskip
Our experiments show that, in the case of highly nonlinear objective functions,  Algorithm 3 performs very well in practice, see Section \ref{sec:Num}.

\section{Examples and Special cases}\label{sec:Exe}

Algorithms \ref{alg:general} and \ref{alg:PJGA} are very general and  encompass a gamut of novel algorithms, each corresponding to various forms of the approximant $P_i$, the error bound function $E_i$, the step-size sequence $\gamma^k$, the block partition, etc. These choices lead to algorithms that can be very different from each other, but \emph{all converging under the same conditions}. These  degrees of freedom offer   a lot of flexibility  to control iteration complexity, communication overhead, and convergence speed.
We outline next several effective  choices for the design parameters  along with  some illustrative examples of  specific algorithms resulting from a proper combination of these choices. 

\noindent \textbf{On the choice of the step-size $\gamma^k$}\emph{.}
An example of step-size rule satisfying conditions i)-iv) in Theorem \ref{Theorem_convergence_inexact_Jacobi}  is:
given $0 < \gamma^{0} \le 1$, let\vspace{-0.2cm}
\begin{equation}\label{eq:gamma}
\gamma^{k}=\gamma^{k-1}\left(1-\theta\,\gamma^{k-1}\right),\quad k=1,\ldots, 
\end{equation}
where $\theta\in(0,1)$ is a given constant. Notice that while this rule
may still require some tuning for optimal behavior, it is quite reliable, since in general we are not
using a (sub)gradient direction, so that many of the well-known practical drawbacks associated
with a (sub)gradient method with diminishing step-size are mitigated in our setting.  Furthermore, this choice of step-size does not require any form of centralized coordination, which is  a favorable feature in a parallel environment. Numerical results in Section \ref{sec:Num} show  the effectiveness of  (the customization of)  \eqref{eq:gamma} on specific problems.
 \begin{remark}[Line-search variants of Algorithms \ref{alg:general} and \ref{alg:PJGA}] \emph{It is   possible to
prove convergence of Algorithms 1 and \ref{alg:PJGA} also using other step-size rules, for example Armijo-like line-search procedures on $V$. 
In fact,  Proposition \ref{Prop_x_y} in Appendix A  shows that the direction $\widehat{\mathbf{x}}\left(\mathbf{x}^{k},\boldsymbol{{\tau}}\right)-\mathbf{x}^{k}$ is a ``good" descent direction.
Based on this result, it is not difficult to prove that if  $\gamma^k$ is chosen by a suitable  
line-search procedure on $V$, convergence of Algorithm 1   to stationary points of Problem (\ref{eq:problem 1}) (in the sense of Theorem 1) is  guaranteed.  
Note that standard line-search methods proposed for smooth functions cannot be applied to $V$ (due to the nonsmooth part $G$); one needs to rely on more sophisticated procedures, e.g., in the spirit of those proposed in \cite{tseng2009coordinate,byrd2013inexact,patriksson1998cost,wright2009sparse}. We provide next  an  example of line-search rule that can be used to compute $\gamma^k$ while guaranteeing convergence of Algorithms 1 and 2 [instead of using rules i)-iii) in Theorem 1]; because of space limitations, we consider only the case of exact solutions (i.e., $\epsilon_i^k=0$ in Algorithm 1 and    $\epsilon_{pi}^k=0$ in Algorithms 2 and 3).
Writing for short $\widehat{\mathbf{x}}_i^k\triangleq \widehat{\mathbf{x}}_i\left(\mathbf{x}^{k},\boldsymbol{{\tau}}\right)$
and  $\Delta \widehat{\x}^k\triangleq (\Delta \widehat{\x}^k_i)_{i=1}^n$, with  $\Delta \widehat{\x}^k_i\triangleq \widehat{\mathbf{x}}_i^k-\mathbf{x}_i^{k}$, and denoting by $(\x)_{\mathcal{S}^k}$ the vector whose component $i$ is equal to $x_i$ if $i\in \mathcal S^k$, and zero otherwise,  we have the following: let $\alpha, \beta\in (0,1)$, choose $\gamma^k=\beta^{\bar{\ell}}$, where  $\bar{\ell}$ is the smallest nonnegative integer $\ell$ such that 
\begin{equation}\nonumber
V\left(\x^k+\beta^{\ell}(\Delta \widehat{\x}^k)_{\mathcal S^k}\right)-V(\x^k) 
\leq -\alpha \cdot \beta^{\ell}\,\|\left(\Delta\widehat{\x}^k\right)_{\mathcal S^k}\|^2.
\end{equation} 
Of course, the above procedure will likely be more efficient in terms of iterations than the one based on diminishing step-size rules [as (\ref{eq:gamma})]. However, performing a line-search on a multicore architecture requires some shared memory and coordination among the cores/processors; therefore we do not consider further this variant.  Finally, we observe that convergence of Algorithms 1 and 2 can also be obtained by choosing a constant (suitably small ) stepsize  $\gamma^k$. This is actually the easiest option, but since it generally leads to extremely slow convergence we omitted   this option from our developments here. }
\end{remark}
\noindent \textbf{On the choice of the error bound function $E_i(\mathbf{x})$}.\\
\noindent $\bullet$
 As we mentioned, the most obvious choice is to take $E_i(\mathbf{x}) =  \|\widehat{\mathbf{x}}_{i}(\x^k,\tau_{i})-\mathbf{x}^k_i\|$.
This is a valuable choice if the computation of $\widehat{\mathbf{x}}_{i}(\x^k,\tau_{i})$ can be  easily accomplished. For instance, in the
LASSO problem with  ${\cal N} = \{1, \ldots, n\}$  (i.e., when each block reduces to a scalar variable),
it is well-known that  $\widehat{\mathbf{x}}_{i}(\x^k,\tau_{i})$ can be computed in closed form using
the soft-thresholding operator \cite{beck_teboulle_jis2009}.\\
\noindent  $\bullet$ In situations where the computation of  $\|\widehat{\mathbf{x}}_{i}(\x^k,\tau_{i})-\mathbf{x}^k_i\|$ is not possible or advisable,
we can resort to estimates. Assume momentarily that $G\equiv 0$. Then it is known \cite[Proposition 6.3.1]{Facchinei-Pang_FVI03} under our assumptions that
$\|\Pi_{X_i} (\x^k_i - \nabla_{\x_i}F(\x^k)) - \x^k_i\|$  is an error bound for the minimization problem in \eqref{eq:decoupled_problem_i} and therefore
satisfies \eqref{eq:error bound}, where $\Pi_{X_i} (\mathbf{y})$ denotes the Euclidean projection of $\mathbf{y}$ onto the closed and convex set $X_i$. In this situation we can choose $E_i(\x^k) = \|\Pi_{X_i} (\x^k_i - \nabla_{\x_i}F(\x^k)) - \x^k_i\|$.
If $G(\x) \not \equiv 0$ things  become more complex. In most cases of practical interest, adequate error bounds can be derived from \cite[Lemma 7]{tseng2009coordinate}.\\
\noindent $\bullet$ It is interesting  to note that the computation of $E_i$ is only needed if a partial update of the (block) variables is
performed. However, an option that is always feasible is to take $S^k = {\cal N}$ at each iteration, i.e.,  update all (block) variables at each iteration.
With this choice we can  dispense with the computation of $E_i$ altogether.\smallskip\\
\noindent \textbf{On the choice of the approximant $P_i(\mathbf{x}_i;\mathbf{x} )$}.\\
\noindent $\bullet$ The most obvious choice for $P_i$  is the linearization of $F$ at $\x^k$ with respect to $\x_i$:
$$ P_i(\x_i; \x^k) = F(\x^k) + \nabla_{\x_i}F(\x^k)^T(\x_i - \x_i^k).$$ With this choice, and taking for simplicity $\Q_i(\x^k) = \mathbf I$, \vspace{-0.1cm}
\begin{equation}\label{eq:proposal 1}
\begin{array}{ll}
 \widehat {\mathbf x}_i(\x^k, \tau_i)=  \underset{\x_{i}\in X_{i}}{\mbox{argmin}\,}  \left \{F(\x^k) + \nabla_{\x_i}F(\x^k)^T(\x_i - \x_i^k)\right. \\ \left.\hspace{3.2cm}+ \dfrac{\tau_i}{2}
\| \x_i- \x_i^k\|^2 + g_i(\x_i)\right\}.\vspace{-0.1cm}
\end{array}
\end{equation}
This is essentially the way   a new iteration is computed in most  (block-)CDMs for the
solution of (group) LASSO problems and its generalizations.\\\noindent $\bullet$ At another extreme we could just take $P_i(\x_i; \x^k) = F(\x_i, \x_{-i}^k)$. Of course, to
have P1 satisfied (cf. Section \ref{sec:Main Results}), we must assume that $F(\x_i, \x_{-i}^k)$ is convex.
 With this choice, and setting  for simplicity $\Q_i(\x^k) =\mathbf  I$, we have
\begin{equation}\label{eq:proposal 2}
 \widehat \x_i(\x^k, \tau_i) \triangleq \underset{\x_{i}\in X_{i}}{\mbox{argmin}\,} \left\{ F(\x_i,  \x^k_{-i}) + \frac{\tau_i}{2}
\| \x_i- \x_i^k\|^2 + g_i(\x_i)\right\},
\end{equation}
thus giving rise to a parallel nonlinear Jacobi type method for the constrained minimization of $V(\x)$.\\
\noindent $\bullet$ Between the two ``extreme'' solutions proposed above, one can consider ``intermediate'' choices. For example,
 If   $F(\x_i, \x_{-i}^k)$ is convex, we can take  $P_i(\x_i;\x^k)$ as a second order approximation of $F(\x_i, \x_{-i}^k)$, i.e.,
\begin{equation}
\begin{array}{ll}
P_i(\x_i; \x^k) = F(\x^k) + \nabla_{\x_i}F(\x^k)^T(\x_i - \x_i^k) \smallskip\\ \hspace{2cm}+ \frac{1}{2} (\x_i - \x_i^k)^T \nabla^2_{\x_i\x_i} F(\x^k) (\x_i - \x_i^k).\end{array}
\end{equation} When  $g_i(\x_i) \equiv 0$, this essentially corresponds to taking a Newton step in minimizing the ``reduced'' problem $\min_{\x_i\in X_i}F(\x_i, \x_{-i}^k)$, resulting in 
\begin{equation}
\begin{array}{lll}
 \widehat \x_i(\x^k, \tau_i) &= \underset{\x_{i}\in X_{i}}{\mbox{argmin}\,}  \left\{ F( \x^k) +\nabla_{\x_i}F( \x^k)^T(\x_i - \x_i^k)\right.\vspace{-0.1cm}
\\&\;
\hspace{1.5cm}+  \left.\frac{1}{2} (\x_i -  \x_i^k)^T \nabla^2_{\x_i\x_i} F( \x^k) (\x_i - \x_i^k)\right. \medskip \\&\;  \hspace{1.5cm}\left. +
 \frac{\tau_i}{2}\| \x_i- \x_i^k\|^2 + g_i(\x_i)\right\}.
 \end{array}
\end{equation}
\noindent $\bullet$ Another ``intermediate'' choice, relying on a specific structure of the objective function that has
important applications is the following. Suppose that $F$ is a sum-utility function, i.e.,
$$
F(\x) = \sum_{j\in J} f_j(\x_i, \x_{-i}),
$$
for some finite set $J$. Assume now that for every $j \in S_i \subseteq J$, the functions $f_j (\bullet, \x_{-i})$ is convex. Then we may set
$$
 P_i(\x_i; \x^k) =  \sum_{j\in S_i}f_j(\x_i, \x_{-i}^k) +  \sum_{j \not \in S_i}\nabla f_j(\x_i, \x_{-i}^k)^T  (\x_i - \x_i^k)
$$
thus   preserving, for each $i$, the favorable convex part of $F$ with respect to $\x_i$ while linearizing the nonconvex parts. This is the
approach adopted in \cite{scutari_facchinei_et_al_tsp13} in the design of multi-users systems, to which we refer for applications in signal processing and communications.

The framework described in Algorithm  \ref{alg:general}  can give rise
to very different schemes, according to the choices one makes for the many
variable features it contains, some of which have been detailed above. Because of space limitation,  we cannot  discuss here all possibilities.
We provide next   just a few instances  of possible algorithms that fall in our framework.
\smallskip

\noindent {\bf Example$\,\#$1$-$(Proximal) Jacobi algorithms for
convex functions.} Consider the simplest problem falling in our setting: the unconstrained minimization of a continuously differentiable convex function, i.e.,
assume that $F$ is convex, $G \equiv 0$, and $X= \Re^n$. Although this is  possibly the best studied problem in nonlinear optimization, classical
parallel methods for this problem \cite[Sec. 3.2.4]{Bertsekas_Book-Parallel-Comp} require very strong contraction conditions. In our framework we can take
$P_i(\x_i;\x^k) = F(\x_i, \x_{-i}^k)$, resulting in a parallel Jacobi-type  method which does not need any additional assumptions. Furthermore our theory shows that
we can even  dispense with the convexity assumption and still get convergence of a Jacobi-type method to a stationary point.
If in addition we take $S^k = {\cal N}$, we obtain the class of methods  studied in
\cite{cohen1978optimization,cohen1980auxiliary,patriksson1998cost,scutari_facchinei_et_al_tsp13}.
\smallskip

\noindent {\bf Example$\,\#$2$-$Parallel coordinate descent methods for LASSO.}
Consider the LASSO problem, i.e., Problem \eqref{eq:problem 1} with $F(\x) = \|\A\x -\bv\|^2 $,  $G(\x)= c \|\x\|_1$, and $X=\Re^n$. Probably, to date,  the most successful class of methods for this problem is
that of CDMs, whereby at each iteration a \emph{single} variable is updated  using \eqref{eq:proposal 1}.
We can easily obtain a parallel version for this method by taking $n_i =1$, $S^k = {\cal N}$ and
 still using  \eqref{eq:proposal 1}. Alternatively, instead of linearizing $F$, we can  better exploit the structure  of $F$ and  use   (\ref{eq:proposal 2}). In fact, it is well known that in LASSO problems subproblem
  (\ref{eq:proposal 2}) can be solved analytically (when $n_i=1$). We  can  easily consider similar methods for the group LASSO problem as well  (just take $n_i >1$).
  \smallskip

\noindent {\bf Example$\,\#$3$-$Parallel coordinate descent methods for Logistic Regression.}
 Consider the Logistic Regression problem, i.e.,  Problem \eqref{eq:problem 1} with
$F(\x) = \sum_{j=1}^m \log(1 + e^{-a_i \y_i^T\x }) $, $G(\x)= c \|\x\|_1$, and
$X=\Re^n$, where $\y_i\in \Re^n$, $a_i\in \{-1, 1\}$, and $c\in \Re_{++}$ are  given
constants.
Since $F(\x_i, \x_{-i}^k)$ is convex, we can take
$ P_i(\x_i; \x^k) = F(\x^k) + \nabla_{\x_i}F(\x^k)^T$ $(\x_i - \x_i^k) + \frac{1}{2} (\x_i - \x_i^k)^T \nabla^2_{\x_i\x_i} F(\x^k) (\x_i - \x_i^k) $ and
thus obtaining a fully distributed and parallel CDM that uses a second order approximation of the smooth function $F$.
Moreover by taking $n_i =1$ and using a soft-thresholding operator, each $\widehat \x_i$ can be computed in closed form.\smallskip

\noindent {\bf Example$\,\#$4$-$Parallel algorithms for dictionary learning problems}. Consider the dictionary learning problem, i.e.,  Problem \eqref{eq:problem 1} with $F(\mathbf{X})=\|\mathbf{Y}-\mathbf{X}_1\mathbf{X}_2\|_F^2$ and $G(\mathbf{X}_2)=c\|\mathbf{X}_2\|_1$, and $X= \{\mathbf{X}\triangleq(\mathbf{X}_1,\mathbf{X}_2)\in \Re^{n\times m} \times  \Re^{m\times N}\,:\, \|\mathbf{X}_{1}\mathbf{e}_i\|^2\leq \alpha_i,\forall i=1,\ldots,m \}$. Since $F(\mathbf{X}_1,\mathbf{X}_2)$ is convex in $\mathbf{X}_1$ and $\mathbf{X}_2$ separately, one can take $P_1(\mathbf{X}_1;\mathbf{X}_2^k)=F(\mathbf{X}_1,\mathbf{X}_2^k)$ and $P_2(\mathbf{X}_2;\mathbf{X}_1^k)=F(\mathbf{X}_1^k,\mathbf{X}_2)$. Although natural, this choice does not lead to a close form solutions of the subproblems associated with the optimization of $\mathbf{X}_1$ and $\mathbf{X}_2$. This desirable property can be obtained using the following alternative approximations of $F$ \cite{MeisamLuoICASSP14}:  $P_1(\mathbf{X}_1;\mathbf{X}^k)=F(\mathbf{X}^k)+\left\langle \nabla_{\mathbf{X_1}} F(\mathbf{X}^k), \mathbf{X}_1 -\mathbf{X}_1^k\right\rangle$ and $P_2(\mathbf{X}_2;\mathbf{X}^k)=F(\mathbf{X}^k)+\left\langle \nabla_{\mathbf{X_2}} F(\mathbf{X}^k), \mathbf{X}_2 -\mathbf{X}_2^k\right\rangle$, where  $\left\langle \mathbf{A},\mathbf{B}\right\rangle\triangleq \text{tr}(\mathbf{A}^T \mathbf{B})$. Note that differently from \cite{MeisamLuoICASSP14}, our algorithm is parallel and more flexible in the choice of the proximal gain terms [cf. \eqref{q_approx}].

\section{Related works}\label{sec:related_works}

 The proposed algorithmic framework  draws on Successive Convex Approximation (SCA) paradigms that have a long history in the optimization literature. Nevertheless, our algorithms and their convergence conditions (cf. Theorems 1 and 2) unify and extend current parallel and sequential SCA methods
in several directions, as outlined next.


\noindent \textbf{(Partially) Parallel  Methods}:  The roots of parallel deterministic SCA schemes (wherein \emph{all} the variables are updated simultaneously)   can be traced back at least to the work of Cohen on the so-called  auxiliary principle
\cite{cohen1978optimization,cohen1980auxiliary} and its related developments, see e.g.
\cite{necoara2013efficient,nesterov2012gradient,
 tseng2009coordinate,beck_teboulle_jis2009,wright2009sparse,bradley2011parallel,
 peng2013parallel,richtarik2012parallel,scutari_facchinei_et_al_tsp13,patriksson1998cost,fukushima1981generalized,mine1981minimization}.
Roughly speaking these works can be divided in two groups, namely: solution methods for \emph{convex} objective functions
\cite{necoara2013efficient,beck_teboulle_jis2009,bradley2011parallel,peng2013parallel,
richtarik2012parallel,cohen1978optimization,cohen1980auxiliary}
and  \emph{nonconvex} ones \cite{nesterov2012gradient,
 tseng2009coordinate,wright2009sparse,scutari_facchinei_et_al_tsp13,
 patriksson1998cost,fukushima1981generalized,mine1981minimization}. All methods in the former group (and \cite{nesterov2012gradient,tseng2009coordinate,wright2009sparse,fukushima1981generalized,mine1981minimization}) are (proximal) gradient  schemes; they thus share the classical drawbacks of gradient-like schemes (cf. Sec. \ref{sec:Introduction}); moreover,  by replacing the convex function $F$ with its first order approximation, they do not take any advantage of any structure of $F$ beyond mere differentiability; exploiting any  available structural properties of $F$, instead,    has been shown to enhance (practical) convergence speed, see e.g. \cite{scutari_facchinei_et_al_tsp13}.  Comparing with the second group of works  \cite{nesterov2012gradient,
 tseng2009coordinate,wright2009sparse,scutari_facchinei_et_al_tsp13,
 patriksson1998cost,fukushima1981generalized,mine1981minimization},
 our algorithmic framework improves on their convergence properties while  adding great  flexibility in the selection of how many   variables  to update at each iteration. For instance, with the exception of \cite{tseng2009coordinate,li2009coordinate,NIPS2011_4425,
 ScherrerTewariHalappanavarHaglin2012parallel_random,peng2013parallel}, all the aforementioned works do not allow  parallel updates of only  a \emph{subset} of all   variables, a feature  that instead, fully explored as we do,  can dramatically  improve the convergence speed of the algorithm, as we show in Section \ref{sec:Num}.  Moreover, with the exception of \cite{patriksson1998cost}, they all require an Armijo-type  line-search, which makes them not appealing for a (parallel)  distributed implementation. A scheme in \cite{patriksson1998cost} is actually based on diminishing step-size-rules, but its convergence properties are quite weak:   not all the limit points of  the sequence generated by this scheme  are guaranteed to be stationary solutions   of (\ref{eq:problem 1}).

Our framework instead i) deals with \emph{nonconvex} (nonsmooth) problems; ii) allows one to use  a much varied array of approximations for $F$ and also inexact solutions of the  subproblems; iii) is fully parallel  and distributed (it does not rely on any line-search); and iv) leads to  the \emph{first distributed convergent} schemes  based on very general  (possibly)  \emph{partial}  updating rules of the optimization variables. 
In fact, among deterministic schemes, we are aware of only three algorithms  \cite{tseng2009coordinate,bradley2011parallel,peng2013parallel} performing at each iteration a parallel update of only  a \emph{subset} of all the variables.  These algorithms however are gradient-like schemes, and do not allow inexact solutions of the subproblems, {when a closed form solution is not available}  (in some large-scale problems the cost of computing the exact  solution of all the subproblems can be prohibitive). In addition,  \cite{tseng2009coordinate} requires an Armijo-type  line-search   whereas  \cite{bradley2011parallel} and \cite{peng2013parallel} are  applicable only to \emph{convex}  objective functions and are not \emph{fully} parallel. In fact,  convergence conditions therein impose a constraint on the maximum number of variables that can be simultaneously updated, a constraint that in many large scale problems is  likely not  satisfied.

\noindent \textbf{Sequential Methods}: Our framework contains as special cases also sequential updates; it is then interesting to  compare our results to sequential schemes too. Given the vast literature on the subject, we consider here only the most recent and general work \cite{razaviyayn2013unified}. 
In \cite{razaviyayn2013unified} the authors consider the minimization of a possibly nonsmooth
function by  Gauss-Seidel methods whereby, at each iteration,
a {\em single} block of variables is updated by minimizing a {\em global upper}
convex approximation of the function. However, finding such an   approximation is generally not an easy task, if not impossible. To cope with this issue,
  the authors  also proposed a variant of their scheme that does not need  this requirement but  uses an Armijo-type  line-search, which however makes the scheme  not suitable for a  parallel/distributed implementation.
Contrary to \cite{razaviyayn2013unified}, in our framework conditions on the approximation
function  (cf. P1-P3) are trivial to be satisfied (in particular, $P$ need not be an upper bound of $F$),
enlarging significantly the class of utility functions $V$ which
the proposed solution method is applicable to.  Furthermore, our framework
 gives rise to parallel and distributed  methods (no line search is used) whose degree of parallelism can be chosen by the user. \vspace{-0.2cm}

\section{Numerical Results}\label{sec:Num}

In this section we provide  some numerical results showing  solid evidence of the viability of our approach.
Our aim is  to compare to state-of-the-art methods as well as  test the influence of  two key features of the proposed algorithmic framework, namely:
parallelism and selective (greedy) selection of (only some) variables to update at each iteration. 
The tests were carried out on i) LASSO and Logistic Regression problems, two of the most studied (convex) instances of Problem (\ref{eq:problem 1}); and ii) \emph{nonconvex} quadratic programming.

All codes have been written in C++. All algebra is performed by using the Intel Math Kernel Library (MKL).
The algorithms were tested on the General Compute Cluster of the Center for Computational Research at the SUNY Buffalo. In particular   we used a partition composed of 372 DELL 16x2.20GHz Intel E5-2660  ``Sandy Bridge''  Xeon  Processor computer nodes  with 128 Gb of DDR4 main memory  and QDR InfiniBand 40Gb/s network card.
In our experiments, parallel algorithms have been tested using up to  40 parallel processes (8 nodes with 5 cores per each), while sequential algorithms ran on a single process (using thus one single core).
\vspace{-0.2cm}

\subsection{LASSO problem}\label{Sec_LASSO}

We implemented the instance of Algorithm 1    described in   Example \# 2 in the previous section, using the approximating function $P_i$
as in \eqref{eq:proposal 2}.
Note that in the case of LASSO problems   $\widehat x_i (\x^k,\tau_i)$, the unique solution  \eqref{eq:proposal 2}, can  be
easily computed in closed form  using the soft-thresholding operator, see \cite{beck_teboulle_jis2009}.\smallskip

\noindent \emph{Tuning of Algorithm 1}:
In the description of our algorithmic framework we considered fixed values of $\tau_i$, but it is clear that varying them a finite number of times does not affect in any way the theoretical convergence properties of the algorithms. We found that the following choices work well in practice: (i) $\tau_i$ are initially all set to $\displaystyle \tau_i
= \text{tr}(\A^T\A)/2n$,  i.e., to half of the mean of the eigenvalues  of $\nabla^2F$;
(ii) all $\tau_i$ are doubled if at a certain iteration the objective function does not decrease; and (iii) they are all halved if the objective function decreases
for ten consecutive iterations or the relative error on the objective function $\texttt{re}(\x)$ is sufficiently small, specifically if
\begin{equation}\label{relative_error}
\texttt{re}(\x) \triangleq  \frac{V(\x)-V^*}{V^*} \le 10^{-2},
\end{equation}
where $V^*$ is the optimal value of the objective function $V$ (in our experiments on LASSO $V^*$ is known).
In order to avoid increments in the objective function, whenever all $\tau_i$ are doubled, the associated iteration is discarded, and in (S.4) of Algorithm 1 it is set  $\x^{k+1} = \x^k$. In any case  we limited  the number of possible updates of the values of $\tau_i$ to 100.

The step-size $\gamma^k$ is updated according  to the following rule:
\begin{equation}\label{gamma_num}
\gamma^{k}=\gamma^{k-1}\left(1-\min\left\{1,\frac{10^{-4}}{\texttt{re}(\x^{k})}\right\} \theta\,\gamma^{k-1}\right),\quad k=1,\ldots,\vspace{-0.1cm}
\end{equation}
with $\gamma^0 = 0.9$ and $\theta = 1e-7$.
The above diminishing rule  is based on  \eqref{eq:gamma} while guaranteeing that $\gamma^k$ does not become too close to zero before the relative error is sufficiently small.

Finally the error bound function is chosen as  $E_i(\x^k) =\|\widehat \x_i (\x^k,\tau_i) - \x_i^k\|$, and  $S^k$  in Step 2 of the algorithm is set to
$$
S^k = \{ i: E_i(\x^k) \geq \sigma M^k\}.
$$
 In our tests we consider two options for $\sigma$, namely: i)  $\sigma=0$, which leads to a \emph{fully parallel} scheme wherein at each iteration \emph{all} variables are updated; and {ii) $\sigma=0.5$, which corresponds to updating only a subset of all the variables at each iteration.}  Note that  for both choices of $\sigma$, the resulting set $S^k$  satisfies the requirement in (S.2)  of Algorithm
 \ref{alg:general}; indeed,   $S^k$ always contains the index $i$ corresponding to the largest $E_i(\x^k)$.
We term  the above instance of   Algorithm 1   \emph{FLEXible parallel Algorithm}  (FLEXA), and  we will refer to these two  versions  as  FLEXA $\sigma =0$ and  FLEXA $\sigma =0.5$. Note that    both   versions  satisfy    conditions of  Theorem \ref{Theorem_convergence_inexact_Jacobi}, and thus are convergent.\smallskip

\noindent \emph{Algorithms in the literature}: We compared our versions of FLEXA
with the most competitive   distributed and sequential   algorithms proposed in the literature to solve the LASSO problem. More specifically, we consider the following schemes.

\noindent $\bullet$ {\bf FISTA}: The Fast Iterative Shrinkage-Thresholding Algorithm (FISTA)
proposed in \cite{beck_teboulle_jis2009} is a first order method and can be regarded as the benchmark algorithm for LASSO
problems.
Building on the separability of the terms in the
objective function $V$, this method can be easily parallelized and thus take advantage of  a parallel architecture. We
implemented the parallel version that use a backtracking procedure to estimate  the   Lipschitz constant $L_F$ of $\nabla F$
\cite{beck_teboulle_jis2009}.

\noindent $\bullet$ {\bf SpaRSA}: This is the first order method proposed in \cite{wright2009sparse};
it is a popular spectral projected gradient method that uses a spectral step
length together with a nonmonotone line search to enhance  convergence.
Also this method can be easily parallelized, which is the version implemented in our tests.
In all the experiments we set the  parameters of  SpaRSA as in \cite{wright2009sparse}: $M=5$, $\sigma=0.01$, $\alpha_{\max} = 1e30$, and $\alpha_{\min} = 1e-30$.

\noindent $\bullet$ {\bf GRock $\&$ Greedy-1BCD}: GRock   is a parallel algorithm proposed in \cite{peng2013parallel} that
performs  well on sparse LASSO problems. We   tested the instance of GRock where  the number of  variables simultaneously
updated  is equal to the number of the parallel processors. It is important to remark that the theoretical convergence
properties of GRock are in jeopardy as the number of  variables updated in parallel increases; roughly speaking,  GRock is
guaranteed to converge  if the columns of the data matrix  $\mathbf A$ in the LASSO problem are ``almost'' orthogonal, a
feature 
that  is not satisfied in many applications. A special instance with convergence guaranteed is the one where only one block
per time (chosen in a greedy fashion) is updated; we refer to this special case as \emph{greedy-1BCD}.

\noindent $\bullet$ {\bf ADMM}: This is a classical Alternating Method of Multipliers (ADMM). We implemented the parallel version as proposed in   \cite{DengLaiPengYin2012scaling}.

In the implementation of the parallel algorithms,  the data matrix $\A$ of the LASSO problem  is generated in a uniform column block distributed manner. More specifically,   each processor generates a slice of the matrix itself such that the overall one can be reconstructed as  $\A = \left[ \A_1 \, \A_2 \, \cdots \, \A_P \right]$, where $P$ is the number of parallel processors, and $\A_i$ has $n/P$ columns for each $i$.
Thus the computation of each product $\A\x$ (which is required to evaluate $\nabla F$) and the norm $\|\x\|_1$ (that is $G$) is divided into
the parallel jobs of computing $\A_i \x_i$ and $\|\x_i\|_1$, followed by a reducing operation.\\
\noindent \emph{Numerical Tests}: We generated six groups of LASSO problems  using the random generator proposed by Nesterov  \cite{nesterov2012gradient}, which permits to control the sparsity of
the solution. For the first five groups, we considered problems with 10,000 variables and matrix $\A$ having 9,000 rows. The five groups differ in the degree of sparsity of the solution; more specifically
the percentage of non zeros in the solution is 1\%, 10\%, 20\%, 30\%, and 40\%, respectively.
The last group is formed by
instances with 100,000 variables and  5000 rows for $\A$, and  solutions having 1\% of non zero variables.
In all experiments and for all the algorithms, the initial point  was set to the zero vector.

Results of our experiments for  the 10,000 variables groups are reported in Fig. \ref{fig}, where we plot the relative
error as defined in \eqref{relative_error} versus the CPU time; all  the curves are obtained using 40 cores, and  averaged over ten independent random
realizations. Note that  the CPU time includes communication times (for distributed algorithms) and the initial time needed
by the methods to perform all pre-iteration computations (this explains why the  curves of   {ADMM} start  after
the others; in fact {ADMM} requires some  nontrivial initializations).   For one   instance, the one corresponding to 1\%  of the sparsity of the solution, we plot also the  relative
error   versus   iterations [Fig. \ref{fig2}(a2)]; similar behaviors of the algorithms have been observed also for the other instances, and thus are not reported.  Results   for the LASSO instance with  100,000 variables  are plotted in Fig. \ref{fig2}. The curves are  averaged over five   random realizations.\\
\begin{figure}[h]
 \vspace{-0.4cm}
\centering
       \begin{subfigure}[]{0.24\textwidth}
                \includegraphics[width=\textwidth]{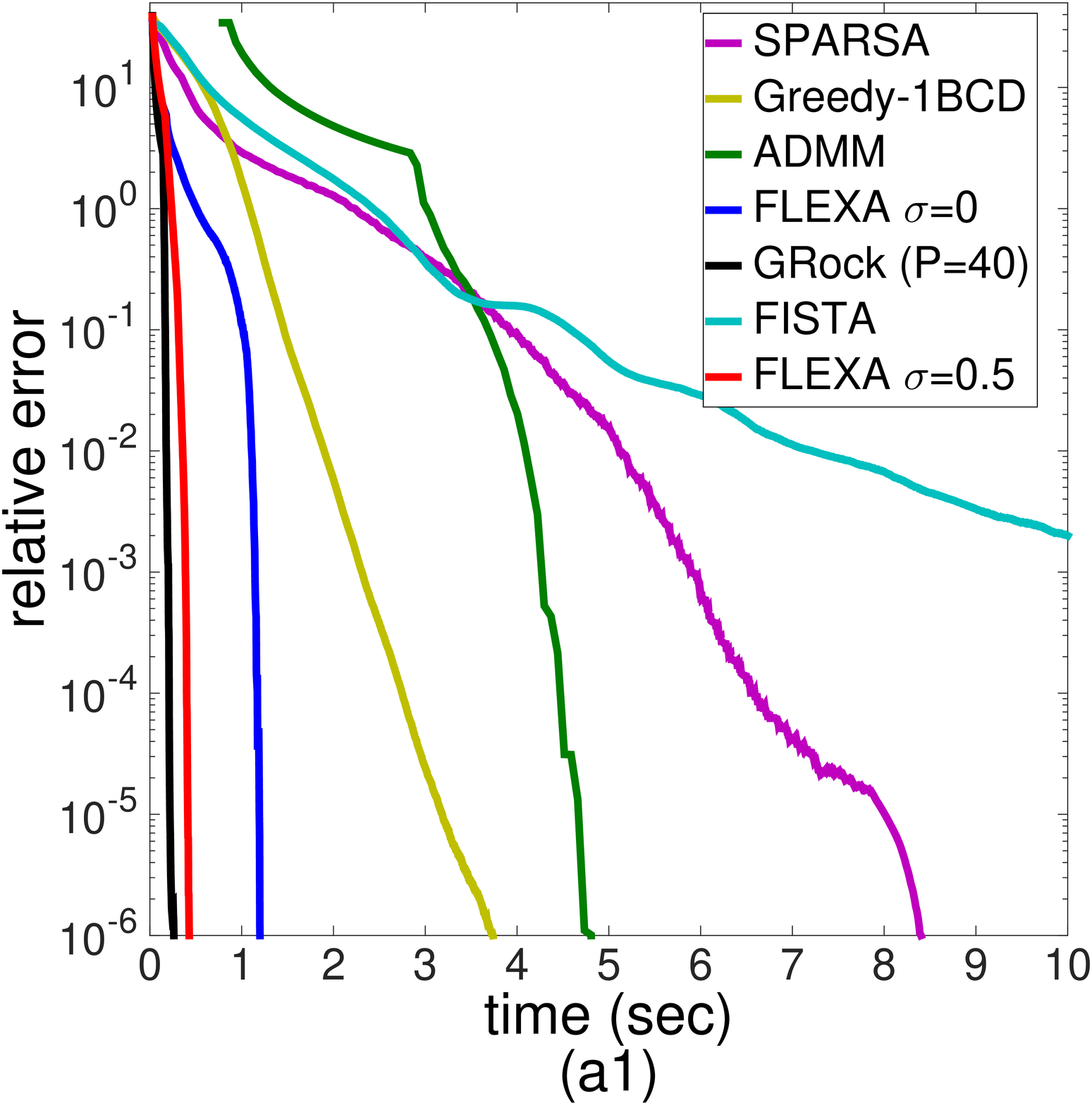}
      \end{subfigure}
      \begin{subfigure}[]{0.24\textwidth}
                \includegraphics[width=\textwidth]{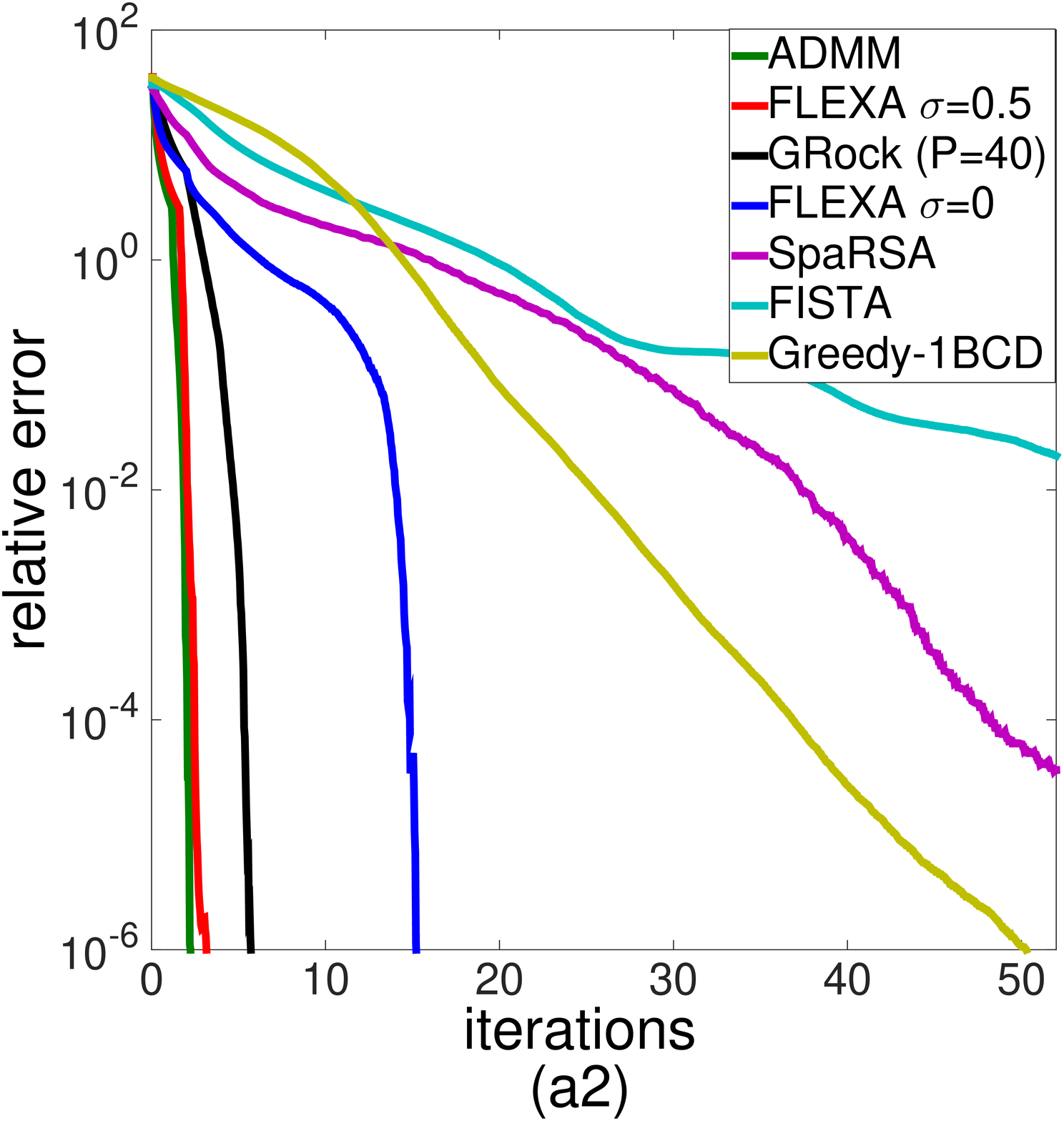}
      \end{subfigure}\vspace{0.1cm}
      \begin{subfigure}[]{0.24\textwidth}
                \includegraphics[width=\textwidth]{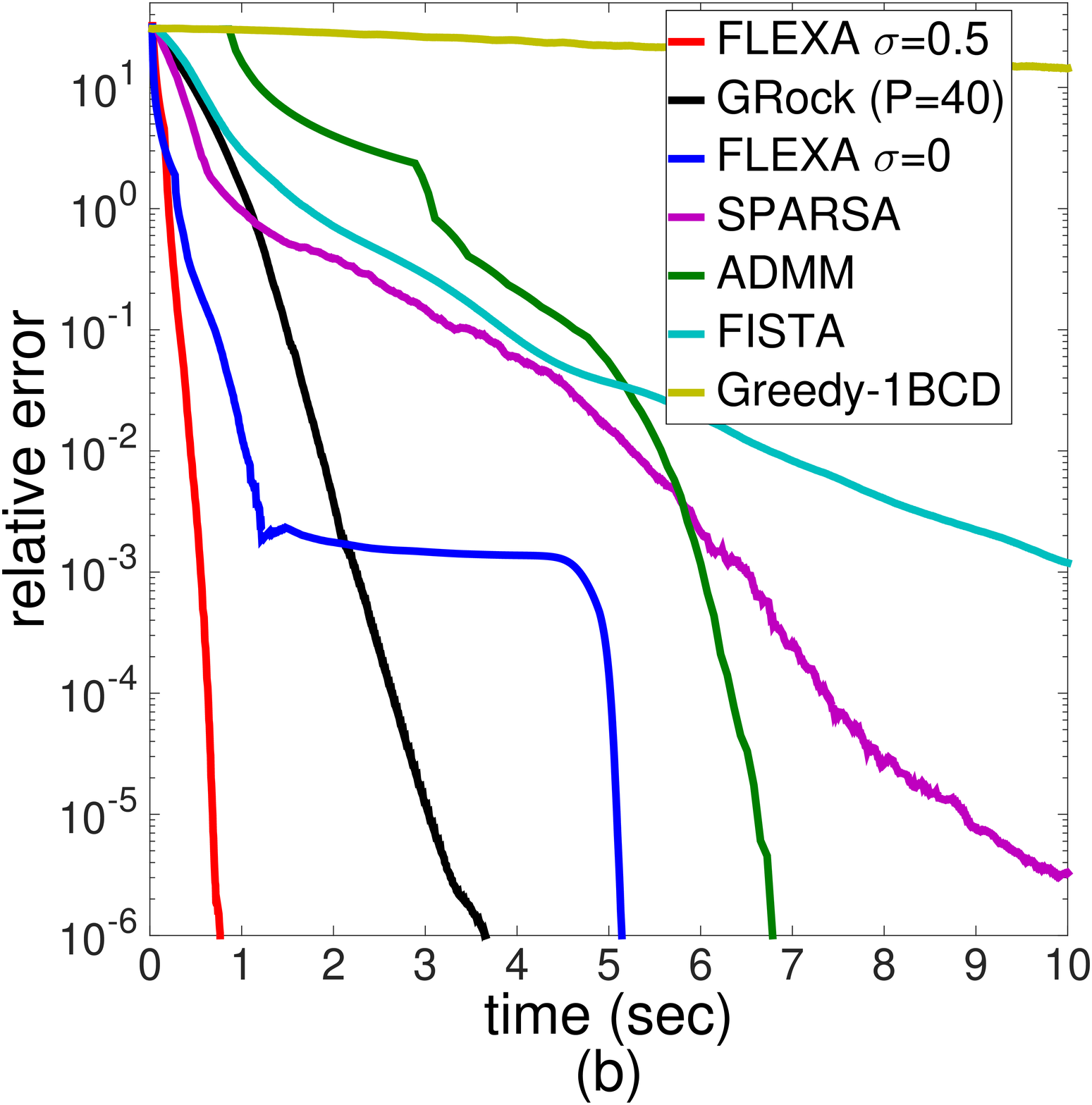}
      \end{subfigure}
      \begin{subfigure}[]{0.24\textwidth}
                \includegraphics[width=\textwidth]{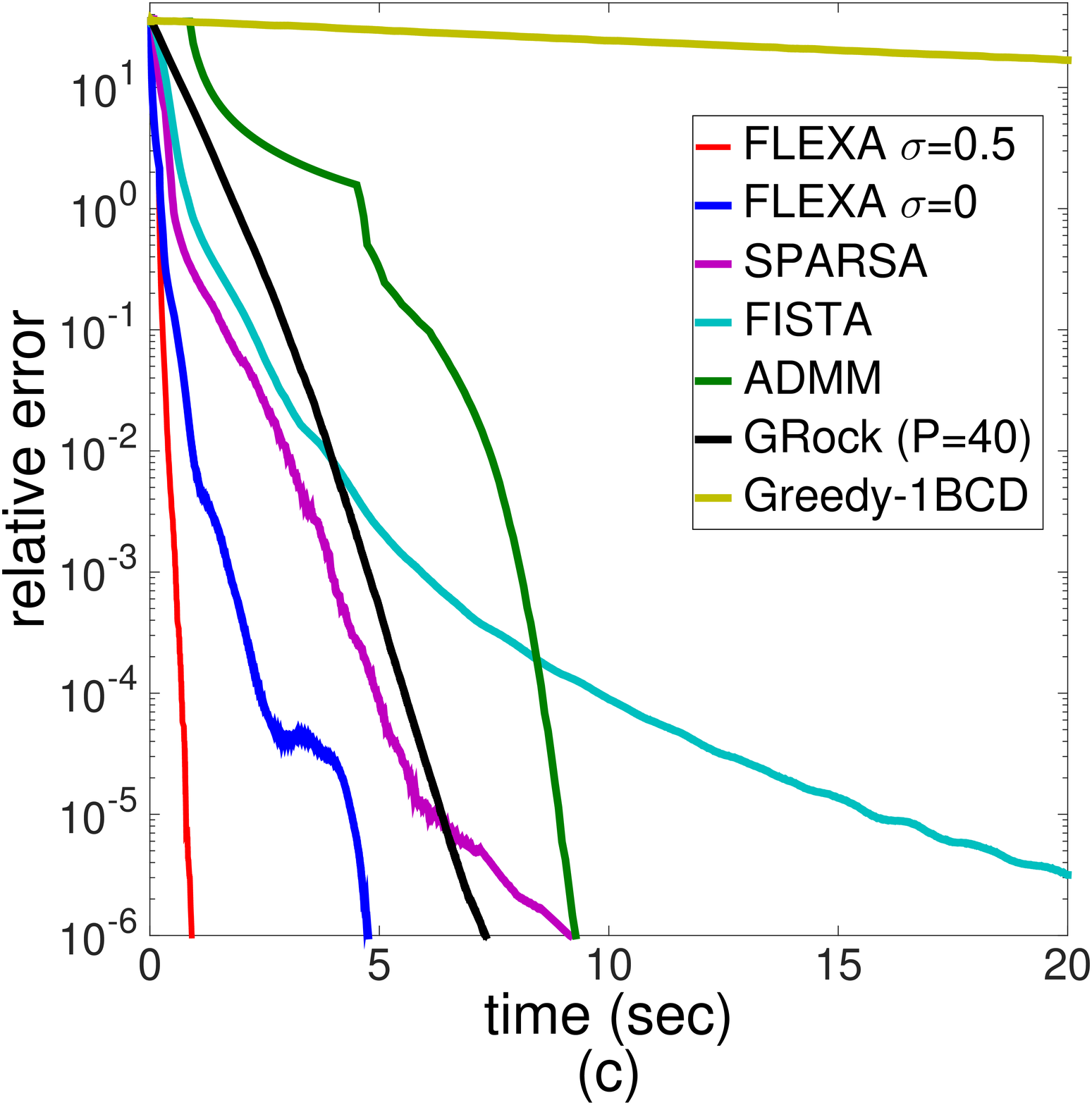}
      \end{subfigure}\vspace{0.0cm}
      \begin{subfigure}[b]{0.24\textwidth}
                \includegraphics[width=\textwidth]{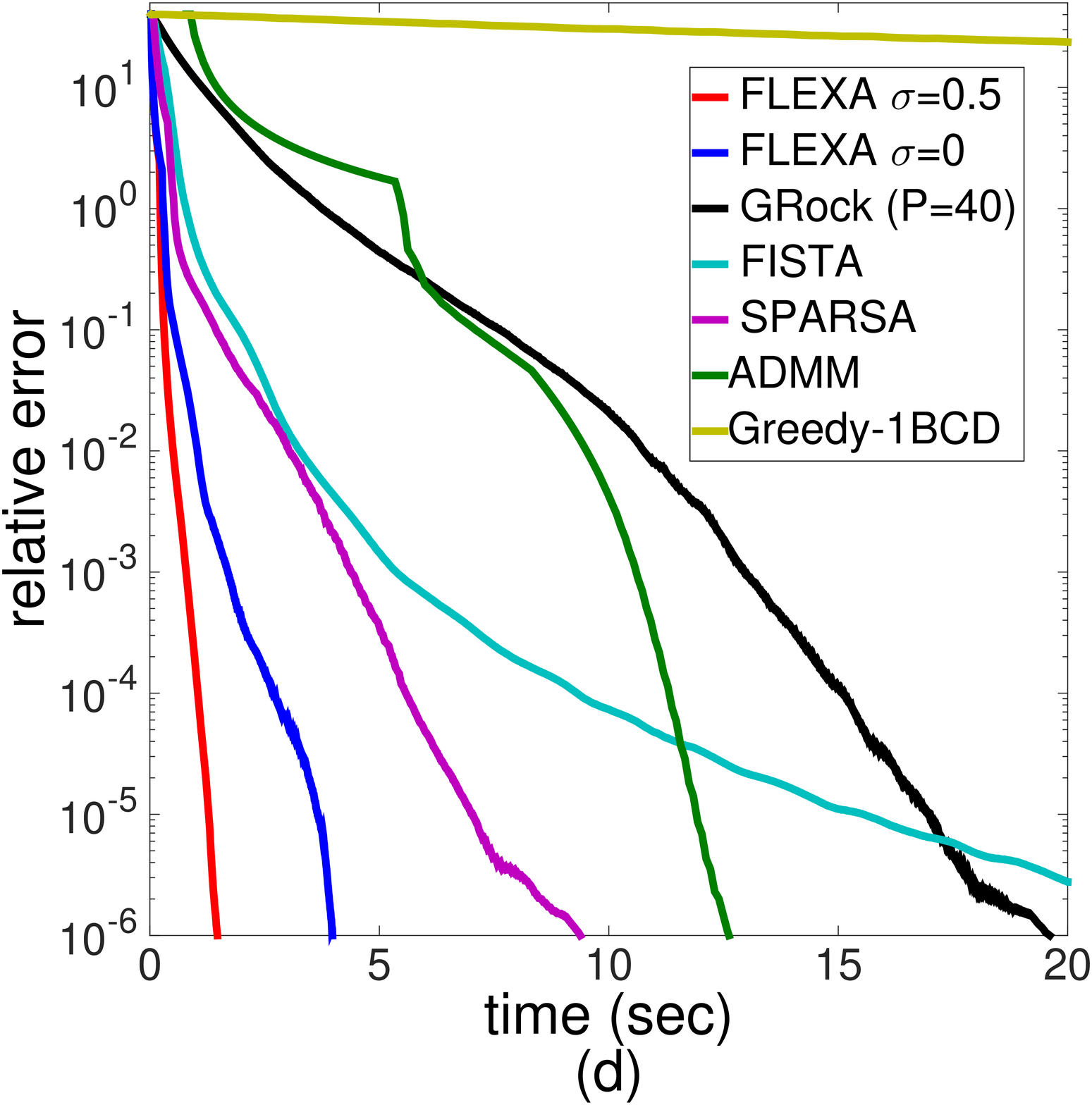}
      \end{subfigure}
      \begin{subfigure}[b]{0.24\textwidth}
                \includegraphics[width=\textwidth]{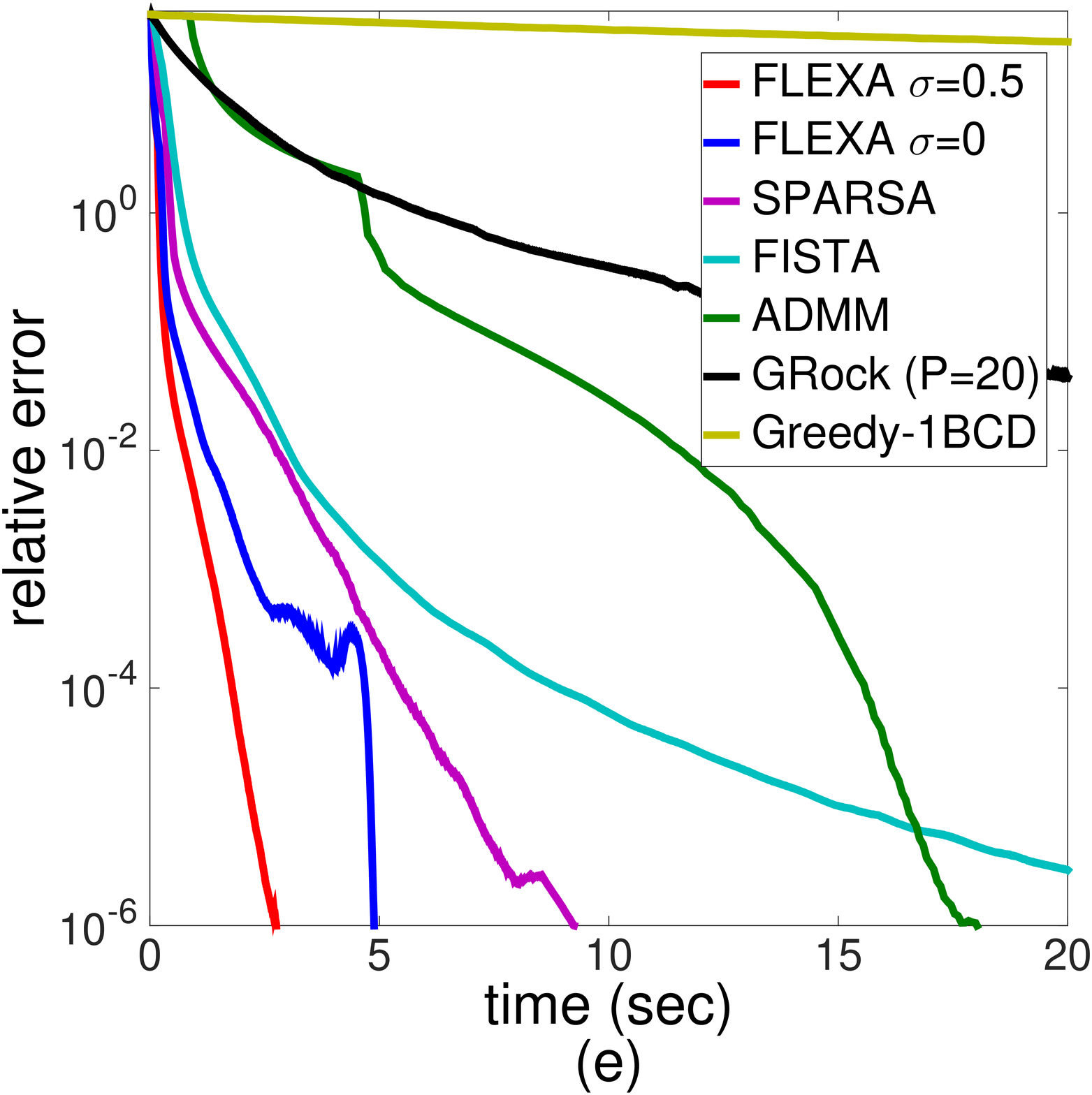}
      \end{subfigure}
\caption{LASSO with 10,000 variables; relative error vs. time (in seconds) for:
(a1) 1\% non zeros - (b) 10\% non zeros - (c) 20\% non zeros - (d) 30\% non zeros - (e) 40\% non zeros; (a2) relative error vs. iterations for 1\% non zeros. \label{fig}}\vspace{-0.4cm}
\end{figure}
\begin{figure}[h]\label{fig2}
\centering
        \begin{subfigure}[ ]{0.27\textwidth}
\hspace*{-0.2cm}
                \includegraphics[width=\textwidth]{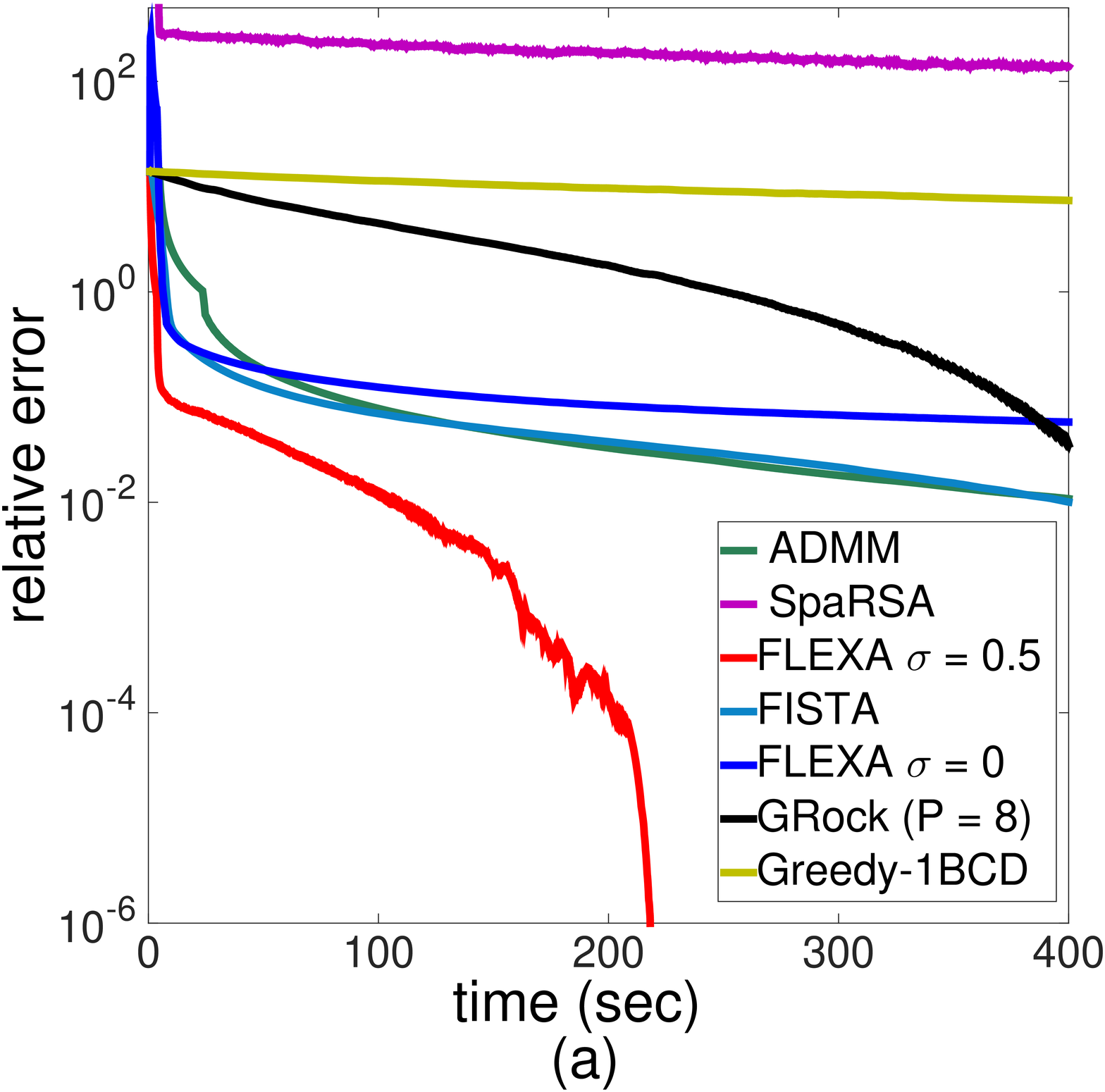}
        \end{subfigure}
        \begin{subfigure}[ ]{0.27\textwidth}
\hspace{-0.4cm}
                \includegraphics[width=\textwidth]{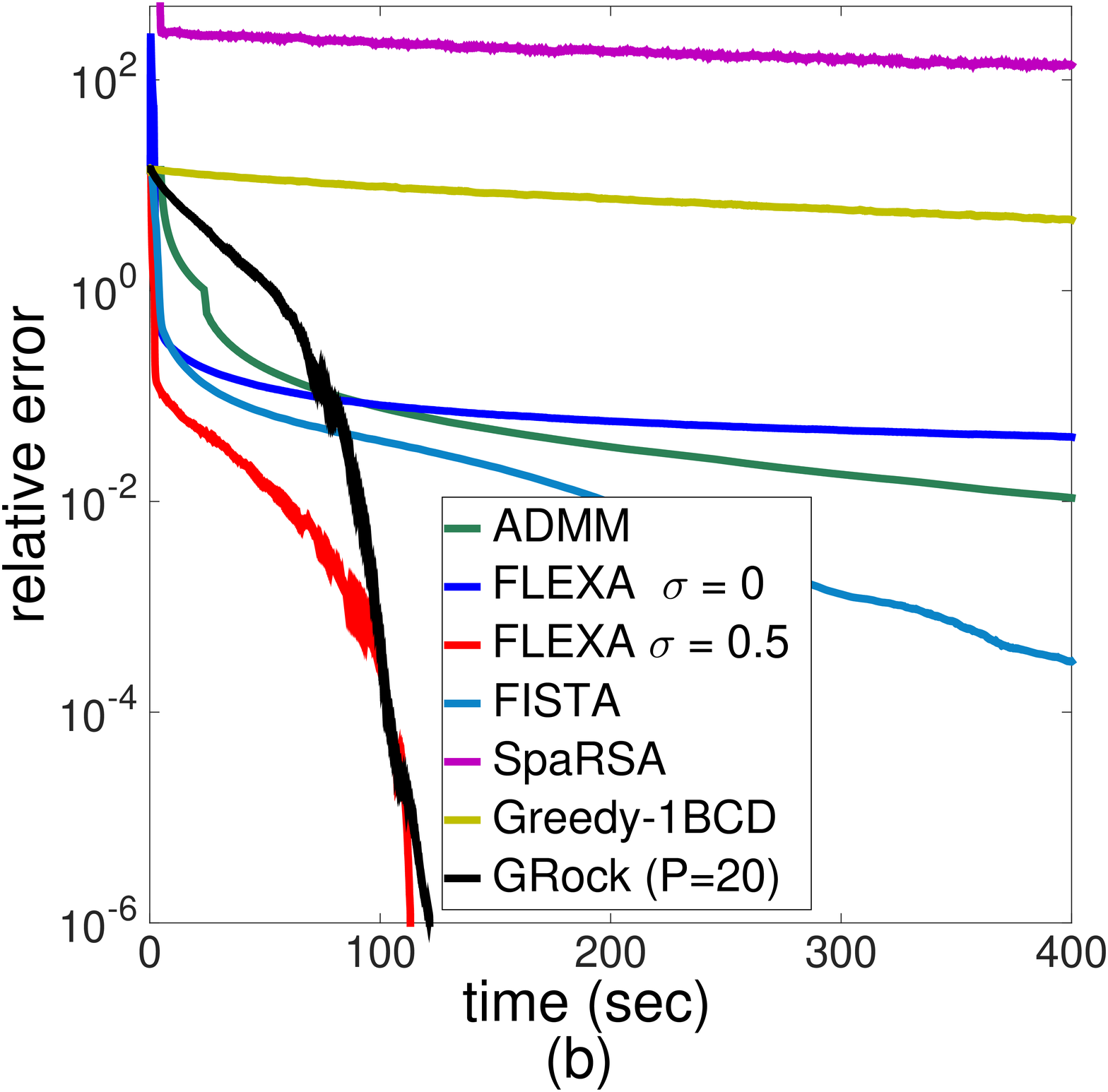}
        \end{subfigure}\vspace{0.0cm}
        \caption{LASSO with $10^5$ variables; Relative error vs. time for: (a) 8 cores - (b) 20 cores.
        \label{fig2}}\vspace{-0.5cm}
        \end{figure}
\indent Given Fig. \ref{fig} and \ref{fig2}, the following comments are in order. On all  the tested problems, { FLEXA}  $\sigma=0.5$ \emph{outperforms in a consistent manner all other implemented algorithms}. 
 In particular, as  the sparsity of the solution decreases, the problems become
harder and  the selective update operated by FLEXA ($\sigma = 0.5$) improves over  FLEXA  ($\sigma = 0$), where instead all variables are updated at each iteration. 
 {FISTA} is capable to approach relatively fast  low accuracy  when the solution is not too sparse, but has difficulties in reaching high  accuracy. {SpaRSA} seems  to be very insensitive to the degree of  sparsity of the solution; it behaves well on 10,000 variables  problems and  not too sparse solutions, but  is much less effective  on very large-scale problems. The version of GRock with $P=40$ is the closest match to {FLEXA}, but only when the problems are very sparse (but it is not supported by a convergence theory on our test problems).
This is consistent with the fact that its convergence properties are at stake when the problems are quite dense.
Furthermore,   if the problem is very large, updating only 40 variables at each iteration, as GRock does, could slow down the convergence, especially when the  optimal solution is not very sparse. From this point of view,   {FLEXA} $\sigma = 0.5$ seems to strike a good balance between not updating  variables that are probably zero at the optimum
and nevertheless update a sizeable amount of variables when needed in order to enhance  convergence.\vspace{-0.2cm}

 \begin{remark}[On parallelism] \emph{ Fig. \ref{fig2} shows that FLEXA seems to exploit well parallelism on LASSO problems.
 Indeed, when passing from 8 to 20 cores, the running time approximately halves. This kind of   behavior has been
 consistently  observed also for  smaller problems and  different number of cores;  because of the space limitation, we do not report these experiments.
  Note that  practical speed-up due to the use of a parallel architecture is given by several factor that are not easily
 predictable and  very dependent on the specific problem at
 hand, including     communication times among the cores, the data format,  etc. In this paper we do not pursue a theoretical study of the speed-up that, given the
 generality of our framework, seems a challenging goal. We finally observe that GRock appears to improve greatly
 with the number of cores. This is due to the fact that in GRock the maximum number of variables that is updated in parallel
 is exactly equal to the number of cores (i.e.,  the degree of parallelism), and this might become a serious drawback on very
 large problems (on top of the fact that convergence is in jeopardy). On the contrary, our theory permits the parallel update of \emph{any} number of variables while guaranteeing convergence.}\vspace{-0.2cm}
 \end{remark}

\begin{remark}[On selective updates] \emph{ It is interesting to comment why FLEXA $\sigma =0.5$ behaves
better than FLEXA   $\sigma =0$. To understand the reason behind this apparently counterintuitive phenomenon,
we first  note that  Algorithm \ref{alg:general} has  the remarkable capability to \emph{identify those
variables that will be zero at a solution}; because of lack of space, we do not provide here the proof of this statement  but
only an informal  description. Roughly speaking,
it can be shown that, for $k$ large enough,  those variables that are zero in $\widehat \x (\x^k,\boldsymbol{\tau})$ will be
zero also in a limiting solution
$\bar \x$. Therefore, suppose that $k$ is large enough so that this identification property  already  takes place (we will say
that  ``we are in the identification phase'') and consider an index $i$ such that
$\bar x_i=0$. Then, if  $x^k_i$ is zero, it is clear, by Steps 3 and 4,  that $x^{k'}_i$ will be zero for all indices $k' > k$,
independently of whether $i$ belongs to $S^k$ or not. In other words, if a variable that is zero at the solution is already zero
when the algorithm  enters the identification phase, {\em that variable will be zero in all subsequent iterations};  this fact, intuitively, should enhance the convergence speed of the algorithm. Conversely, if when we enter the identification phase $x^k_i$ is not zero, the algorithm will have to bring it back to zero iteratively. It should then be clear why updating only variables that we have ``strong'' reason to believe will be non zero at a solution is a better strategy than updating them all. Of course, there may be a problem dependence and the best value of $\sigma$ can vary from problem to problem. But we believe that the explanation outlined above gives  firm theoretical ground to the idea that it might be wise to ``waste" some calculations and perform only a partial update of the variables.}\hfill $\square$\end{remark} \vspace{-0.5cm}

%
%

\subsection{Logistic regression problems}

The logistic regression problem is described in Example \#3  (cf. Section III) and is a highly nonlinear problem involving
many exponentials that, notoriously, give rise to  numerical difficulties. Because of these high nonlinearities,   a Gauss-Seidel approach is expected to be  more effective than a pure Jacobi method, a fact that  was confirmed by our  preliminary tests. For this reason,
for the logistic regression problem we tested  also an instance of Algorithm \ref{alg:PJGA  bis}; we term it  GJ-FLEXA. The setting of the free parameters in GJ-FLEXA is essentially the same described for LASSO, but with the following differences:
\begin{description}
\item[ \rm (a)] The approximant  $P_i$ is chosen as the second order approximation of the original function $F$ (Ex. \#3, Sec. III);
\item[ \rm (b)] The initial $\tau_i$ are set to $\text{tr}(\mathbf{Y}^T\mathbf{Y})/2n$ for all $i$, where $n$ is the total number of variables and $\mathbf{Y} = [ \y_1 \, \y_2 \, \cdots \, \y_m ]\trt. $
\item[ \rm (c)] Since the optimal value $V^*$ is   not known for the logistic regression problem,  we no longer
use $\texttt{re}(\x)$ as merit function but
$\|\mathbf{Z}(\x)\|_\infty$, with
$
\mathbf{Z}(\x) = \nabla F(\x) - \Pi_{[-c,c]^n} \left(\nabla F(\x) - \x\right).
$
Here the projection  $\Pi_{[-c,c]^n} (\mathbf{z})$ can be  efficiently computed; it acts component-wise on $\mathbf{z}$,
since  $[-c,c]^n=[-c,c]\times \cdots \times [-c,c]$.
Note  that  $Z(\x)$ is a valid optimality measure function; indeed,  $\mathbf{Z}(\x) = \mathbf{0}$
is equivalent to the standard necessary optimality condition for Problem \eqref{eq:problem 1}, see \cite{byrd2013inexact}.
Therefore, whenever $\texttt{re}(\x)$ was used for the Lasso problems, we now use $\|\mathbf{Z}(\x)\|_\infty$ [including in the step-
size rule \eqref{gamma_num}].
\end{description}

 We tested the algorithms on   three  instances of the logistic regression problem that are widely used in the literature, and  whose  essential data  features are given in Table \ref{Tab:LogReg};
 we downloaded the data from  the LIBSVM repository {\tt http://www.csie.ntu.edu.tw/}$\sim${\tt cjlin/libsvm/}, which we
 refer to for a detailed description of the test problems. In our implementation, the matrix $\mathbf{Y}$ is stored in a
 column block distributed manner $\mathbf{Y} = \left[ \mathbf{Y}_1 \, \mathbf{Y}_2 \, \cdots \, \mathbf{Y}_P \right]$,
 where $P$ is the number of parallel processors.  We compared FLEXA ($\sigma = 0.5)$ and GJ-FLEXA with the other parallel algorithms (whose tuning of the free parameters  is the same as in  Fig.
 \ref{fig} and Fig. \ref{fig2}), namely: {FISTA}, {SpaRSA}, and {GRock}.
For the logistic regression problem, we also tested one more algorithm, that we call CDM. This Coordinate Descent Method is an extremely efficient   Gauss-Seidel-type method (customized for logistic regression), and is part of the LIBLINEAR package  available at \texttt{http://www.csie.ntu.edu.tw/$\sim$cjlin/}.

In Fig. \ref{fig3}, we plotted the relative error vs. the CPU time (the latter defined as in  Fig. \ref{fig} and Fig. \ref{fig2}) achieved by the  aforementioned algorithms for the three datasets, and using a different number of cores, namely: 8, 16, 20, 40; for each algorithm but GJ-{FLEXA}  we report only the best performance over  the aforementioned numbers of cores.   Note that   in order to plot the relative error, we had to  preliminary estimate  $V^*$ (which  is not known  for logistic regression problems). To  do so,   we  ran GJ-{FLEXA}   until the merit function value $\|Z(\x^k)\|_\infty$ went below $10^{-7}$, and used the corresponding value of the objective function as estimate of $V^*$. We remark that we used this value only to plot the curves.   Next to each plot,   we also reported the overall FLOPS counted up till  reaching the  relative errors as indicated in the table. Note that the FLOPS  of GRock on real-sim and rcv1 are those counted  in 24 hours  simulation time; when terminated,  the algorithm achieved a relative error that was still very far from the reference values set in our experiment. Specifically, GRock reached $1.16$ (instead of $1e-4$) on real-sim and $0.58$ (instead of $1e-3$) on rcv1; the counted FLOPS up till those error values   are still reported in the tables.

\begin{table}[h]
 \centering 
 \scriptsize
 \begin{tabular}{|r| c c c|}
\hline
  Data set & $m$ & $n$ & $c$ \\ \hline
 \textsf{gisette (scaled)} & $6000$ & $5000$ & $0.25$ \\
 \textsf{real-sim} & $72309$ & $20958$ & $4$ \\
 \textsf{rcv1}   & $677399$ & $47236$ & $4$ \\
 \hline
 \end{tabular}\vspace{-0.1cm}
 \caption{Data sets  for logistic regression tests \label{Tab:LogReg} } \vspace{-0.2cm}
\end{table}

\begin{figure}[h]
\hspace{-0.5cm}
 \vspace{-0.3cm}

       \begin{subfigure}[]{0.22\textwidth}
                \includegraphics[width=\textwidth]{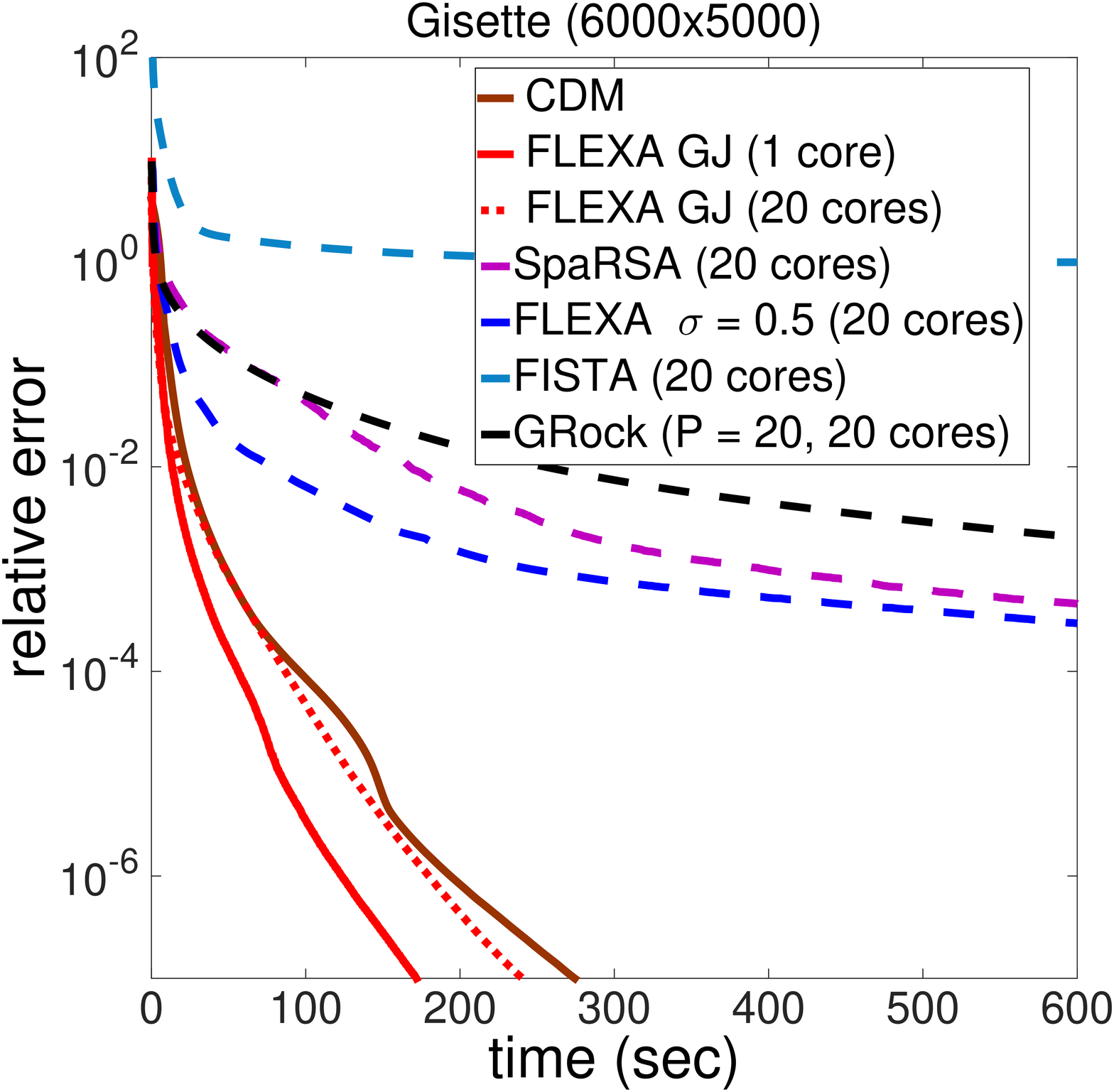}
      \end{subfigure}
      \hspace{-0.2cm}
      \begin{subfigure}[]{0.19\textwidth}
    \scriptsize  \begin{tabular}{|c|c|}
\hline
Algo.  & FLOPS (1e-2/1e-6)\tabularnewline
\hline
GJ-FLEXA (1C) & 1.30e+10/1.23e+11\tabularnewline
\hline
GJ-FLEXA (20C) & 5.18e+11/5.07e+12\tabularnewline
\hline
FLEXA $\!\sigma\!=\!0.5$ (20C) & 1.88e+12/4.06e+13\tabularnewline
\hline
CDM & 2.15e+10/1.68e+11\tabularnewline
\hline
SpaRSA (20C) & 2.20e+12/5.37e+13\tabularnewline
\hline
FISTA (20C) & 3.99e+12/5.66e+13\tabularnewline
\hline
GroCK (20C) & 7.18e+12/1.81e+14\tabularnewline
\hline
\end{tabular}
\end{subfigure}\vspace{0.1cm}\hspace{-0.5cm}

      \begin{subfigure}[]{0.22\textwidth}
                \includegraphics[width=\textwidth]{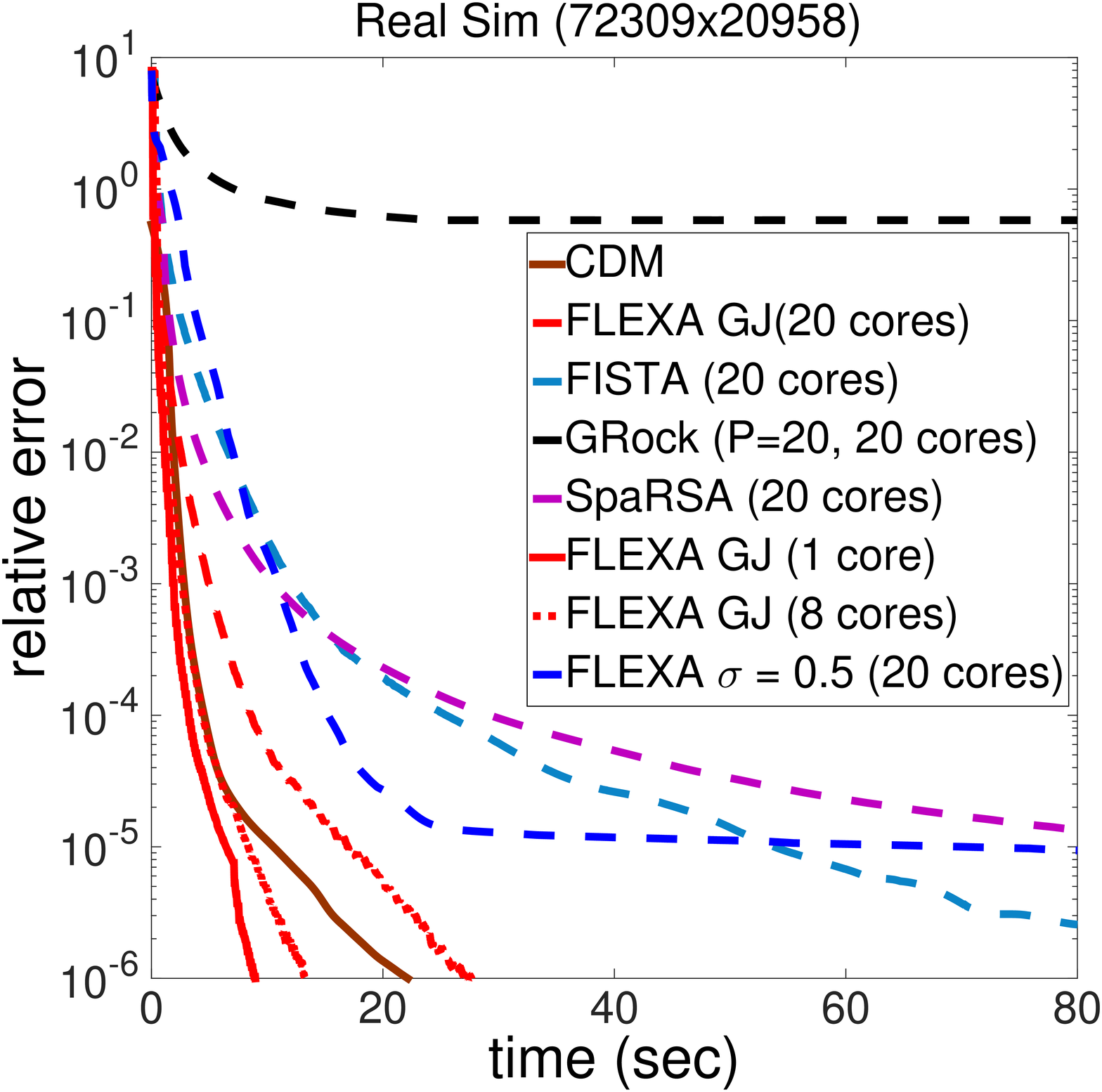}
      \end{subfigure}
            \hspace{-0.4cm}
       \begin{subfigure}[]{0.19\textwidth}
              \scriptsize  \begin{tabular}{|c|c|}
\hline
 Algorithms & FLOPS (1e-4/1e-6)\tabularnewline
\hline
GJ-FLEXA  (1C) & 2.76e+9/6.60e+9\tabularnewline
\hline
GJ-FLEXA (20C) & 9.83e+10/2.85e+11\tabularnewline
\hline
FLEXA $\!\sigma\!=\!0.5$  (20C) & 3.54e+10/4.69e+11\tabularnewline
\hline
CDM & 4.43e+9/2.18e+10\tabularnewline
\hline
SpaRSA (20C) & 7.18e+9/1.94e+11\tabularnewline
\hline
FISTA (20C) & 3.91e+10/1.56e+11\tabularnewline
\hline
GroCK (20C) & 8.30e+14 (after 24h)
 \tabularnewline
\hline
\end{tabular}
      \end{subfigure}\vspace{0.1cm}
      
      \begin{subfigure}[]{0.22\textwidth}
                \includegraphics[width=\textwidth]{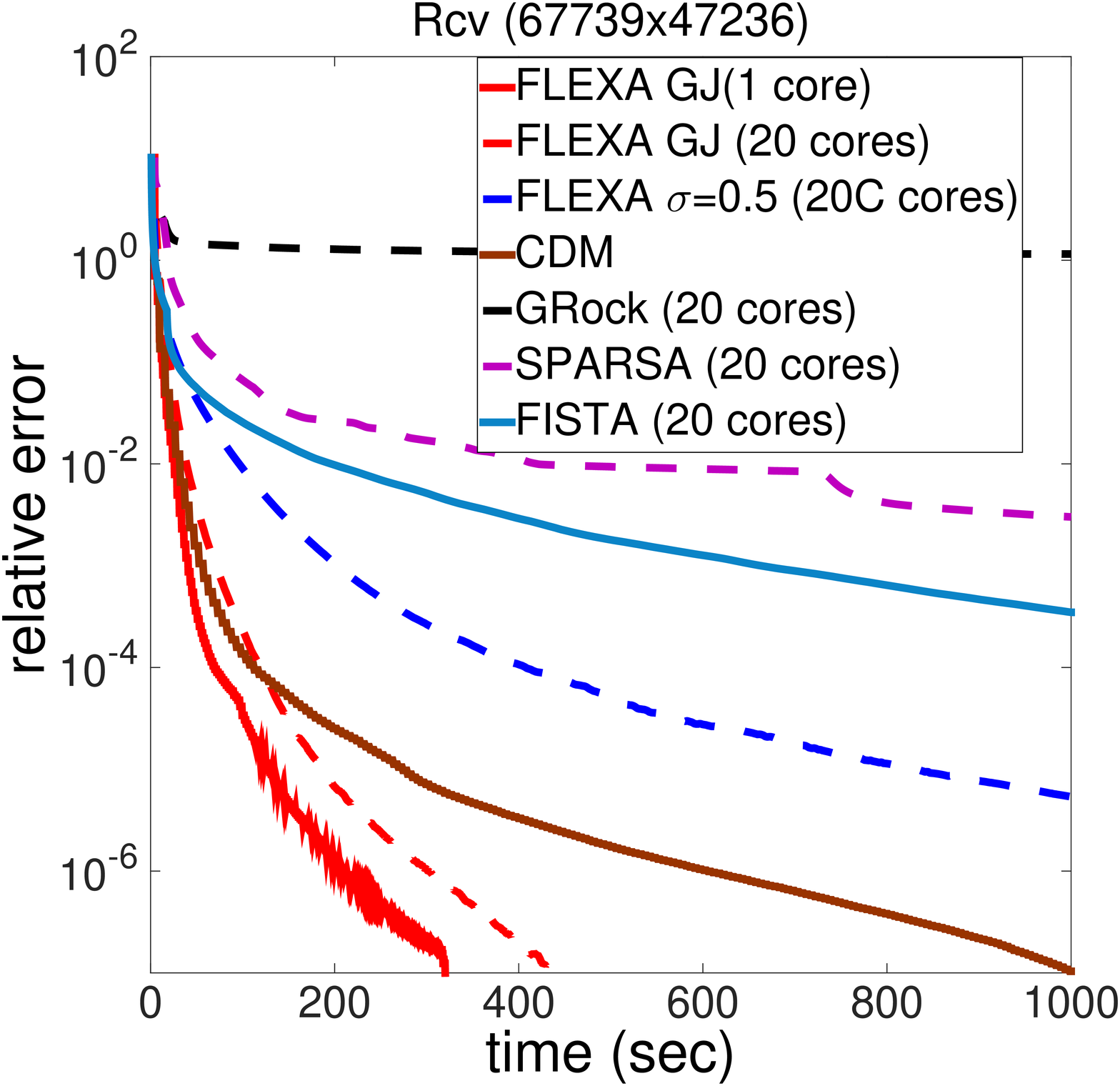}
      \end{subfigure}
            \hspace{-0.4cm}
       \begin{subfigure}[]{0.19\textwidth}
        \scriptsize     \begin{tabular}{|c|c|}
\hline
Algorithms  & FLOPS (1e-3/1e-6)\tabularnewline
\hline
GJ-FLEXA  (1C) & 3.61e+10/2.43e+11\tabularnewline
\hline
GJ-FLEXA  (20C) & 1.35e+12/6.22e+12\tabularnewline
\hline
FLEXA $\!\sigma\!=\!0.5$  (20C) & 8.53e+11/7.19e+12\tabularnewline
\hline
CDM & 5.60e+10/6.00e+11\tabularnewline
\hline
SpaRSA (20C) & 9.38e+12/7.20e+13\tabularnewline
\hline
FISTA (20C) & 2.58e+12/2.76e+13\tabularnewline
\hline
GroCK (20C) & 1.72e+15 (after 24h) \tabularnewline
\hline
\end{tabular}\end{subfigure}
\caption{Logistic Regression problems: Relative error vs. time (in seconds) and FLOPS  for
i) gisette,  ii) real-sim, and iii) rcv.  \label{fig3}}\vspace{-0.6cm}
\end{figure}

The analysis of the figures shows that,  due to the high nonlinearities of the objective function, the more performing methods  are  the 
Gauss-Seidel-type ones.
In spite of this, FLEXA still behaves  quite well. But GJ-FLEXA with one core, thus a non parallel method, clearly outperforms all other 
algorithms. 
The explanation can be the following.  GJ-FLEXA with one core is essentially a Gauss-Seidel-type method but with two key differences: the 
use of a  stepsize and more importantly a (greedy) selection rule by which only  some variables are updated at each round. As the number 
of cores   increases, the algorithm gets ``closer and closer'' to a Jacobi-type method, and because of the high nonlinearities, moving along 
a ``Jacobi direction''  does not bring improvements. In conclusion, for logistic regression problems, our  experiments suggests   that while 
the (opportunistic) selection of variables to update seems useful and brings to improvements even in comparison to the extremely efficient, 
dedicated CDM algorithm/software, parallelism (at least, in the form  embedded in our scheme), does not appear to be beneficial as instead 
observed for LASSO problems.

\subsection{Nonconvex quadratic problems}
We consider now a nonconvex instance of Problem (\ref{eq:problem 1}); because of space limitation, we present   some   preliminary results only on non convex quadratic problems; a more detailed analysis of the  behavior of our method in the nonconvex case is  an important topic that  is the subject of a forthcoming paper. Nevertheless, current results   suggest that FLEXA's  good behavior on convex problems extends to the nonconvex setting.

Consider the following nonconvex instance of Problem (\ref{eq:problem 1}): 
\begin{equation}\label{eq:problem_ncvx}
{\displaystyle \begin{array}{cl}
{\displaystyle {\min_{\mathbf{x}}}} & \underset{F(\mathbf{x})}{\underbrace{\|\A\x-\bv\|^{2}-\bar{c}\|\x\|^{2}}}+\underset{G(\mathbf{x})}{\underbrace{c\|\x\|_{1}}}\\
\mbox{s.t.} & -b\leq x_{i}\leq b,\quad\forall i=1,\ldots,n,
\end{array}}
\end{equation}
where $\bar{c}$ is a positive constant chosen so that $F(\mathbf{x})$ is no longer convex. We simulated two instances of (\ref{eq:problem_ncvx}), namely: 1) $\A\in \mathbb{R}^{9000\times 10000}$ generated using   Nesterov's model (as in LASSO problems in Sec. \ref{Sec_LASSO}), with  
1\% of nonzero in the solution, $b=1$, $c=100$, and $\bar c=1000$; and  2) $\A\in \mathbb{R}^{9000\times 10000}$ as in 1) but with 10\% sparsity, $b=0.1$,  $c=100$, and $\bar c=2800$.  
Note that the Hessian of  $F$ in the two aforementioned cases has the same eigenvalues of the Hessian of $\|\A\x-\bv\|^{2}$ in the original LASSO problem, but translated to the left by
$2\bar c$. In particular,  $F$ in 1) and 2)  has  (many) minimum 
eigenvalues equal to  $-2000$ and $-5600$, respectively; therefore, the objective function $V$ in (\ref{eq:problem_ncvx}) is (markedly) nonconvex. Since $V$ is    now always unbounded from below by construction, we   added in (\ref{eq:problem_ncvx})   box constraints.

\noindent \emph{Tuning of Algorithm 1}: We ran FLEXA using the same tuning as for LASSO problems in Sec.  \ref{Sec_LASSO}, but adding the extra condition  $\tau_i> \bar c$, for all $i$, so that the resulting    one dimensional subproblems \eqref{eq:decoupled_problem_i} are convex and  
can be solved in closed form (as for LASSO).
As termination criterion we used the merit function $\|\bar{\mathbf{Z}}(\x)\|_\infty$, with
$\bar{\mathbf{Z}}(\x)\triangleq [\bar{{Z}}_1(\x),\ldots, \bar{{Z}}_n(\x)]$, and
$$\bar{Z}_i(\x) \triangleq \left\{ \begin{array}{ll}   0 & \text{ if } Z_i(\x)\leq 0 \text{ and } x_i=b,\\ 0 & \text{ if } Z_i(\x) \geq 0 \text{ and } x_i=-b,\\ Z_i(\x) & \text{
otherwise,} \end{array}\right.
$$
where $\mathbf{Z}(\x) \triangleq \nabla F(\x) - \Pi_{[-b,b]^n} \left(\nabla F(\x) - \x\right),$ and $Z_i(\x)$ denotes the $i$-th component of $\mathbf{Z}(\x)$.
We stopped the iterations when $\|\bar{\mathbf{Z}}(\x)\|_\infty \le 1e-3$

We compared  FLEXA  with FISTA and SpaRSA. Note that among  all algorithms considered in the previous sections, only SpaRSA has convergence guarantees for nonconvex problems; but we also added FISTA to the
comparison because it seems to perform well in practice and because of its benchmark status in the convex case.
On our tests, the three algorithms always converge  to the same (stationary) point.
Computed stationary solutions of class 1) of simulated  problems have approximately 1\% of non
zeros and  0.1\% of variables on the bounds, whereas  those of  class 2) have 3\% of non zeros and 0.3\% of variables on the bounds.
Results of our experiments for the 1\% sparsity problem are reported in Fig. \ref{figN1} and those for the 10\% one in Fig. \ref{figN2}; we
plot the relative
error as defined in \eqref{relative_error} and  the merit value versus the CPU time; all  the curves are obtained using 20 cores and averaged over 10 independent random realization. The CPU time
includes communication times and the initial time needed
by the methods to perform all pre-iteration computations.
\begin{figure}[h]
 \vspace{-0.2cm}
\centering
       \begin{subfigure}[]{0.24\textwidth}
                \includegraphics[width=\textwidth]{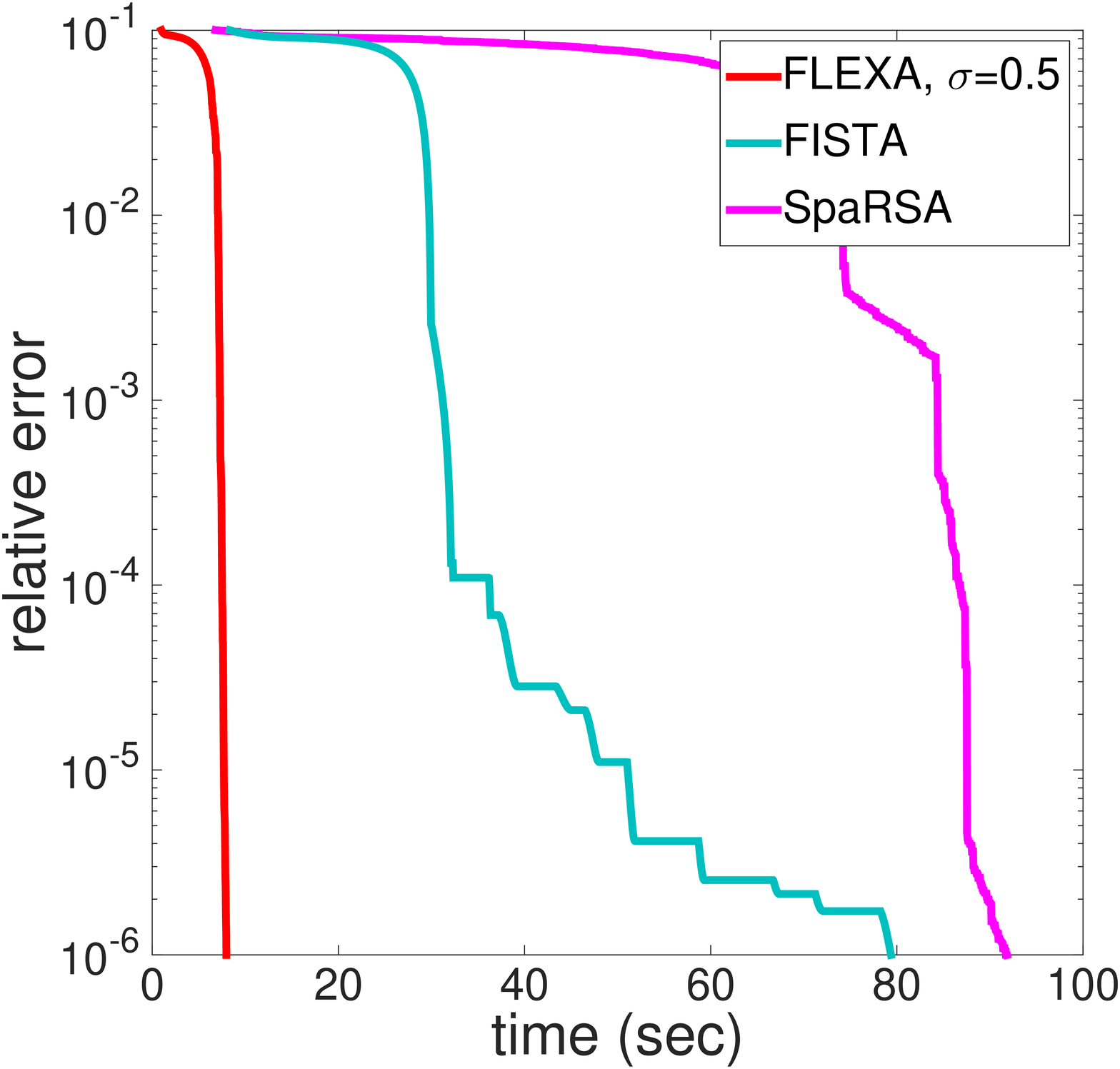}
      \end{subfigure}
      \begin{subfigure}[]{0.24\textwidth}
                \includegraphics[width=\textwidth]{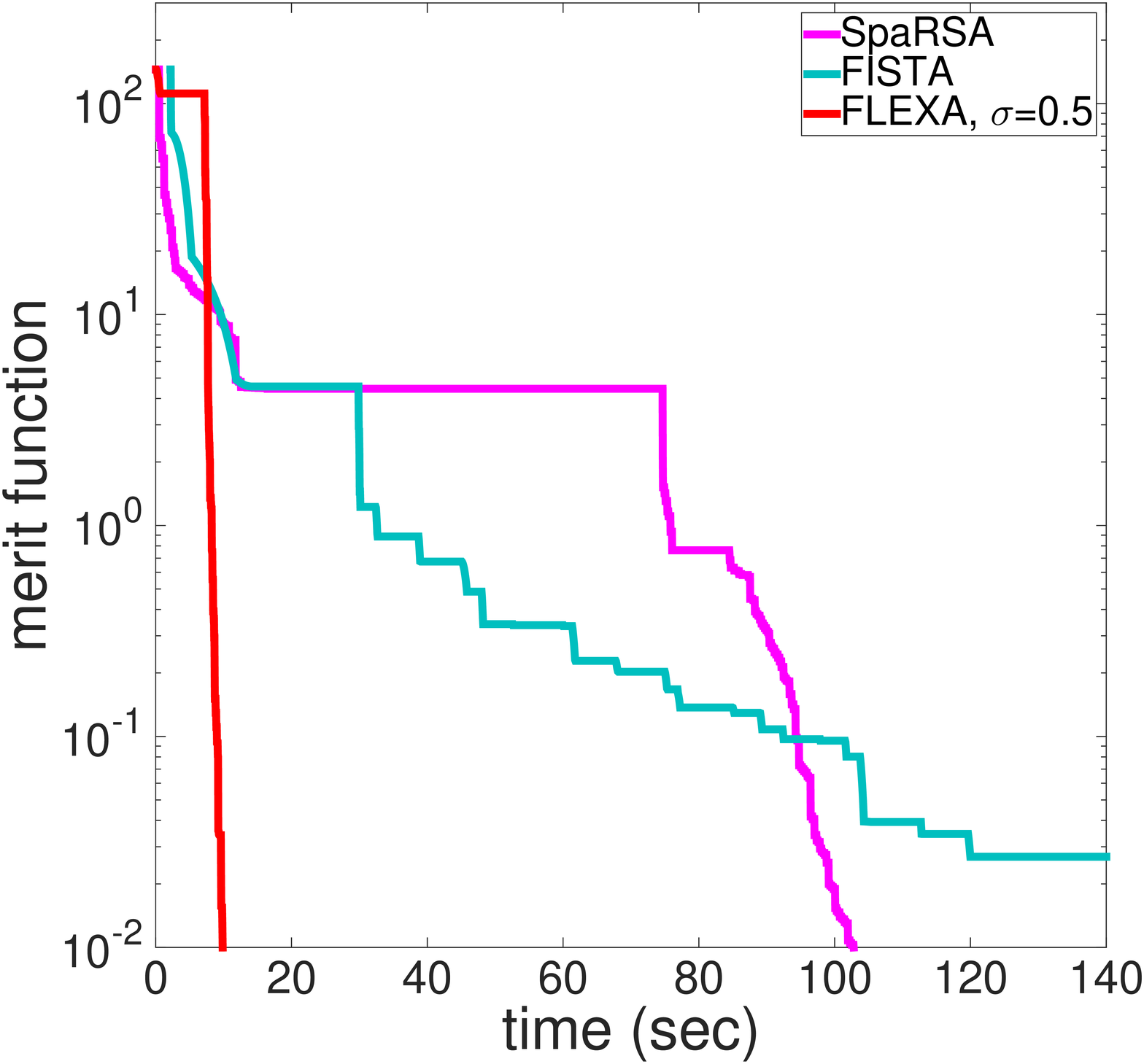}
      \end{subfigure}
\caption{Nonconvex problem (\ref{eq:problem_ncvx}) with 1\% solution sparsity: relative error vs. time (in seconds), and merit value vs. time (in seconds). \label{figN1}}  \vspace{-0.2cm}
\end{figure}
\begin{figure}[h]
\centering
       \begin{subfigure}[]{0.24\textwidth}
                \includegraphics[width=\textwidth]{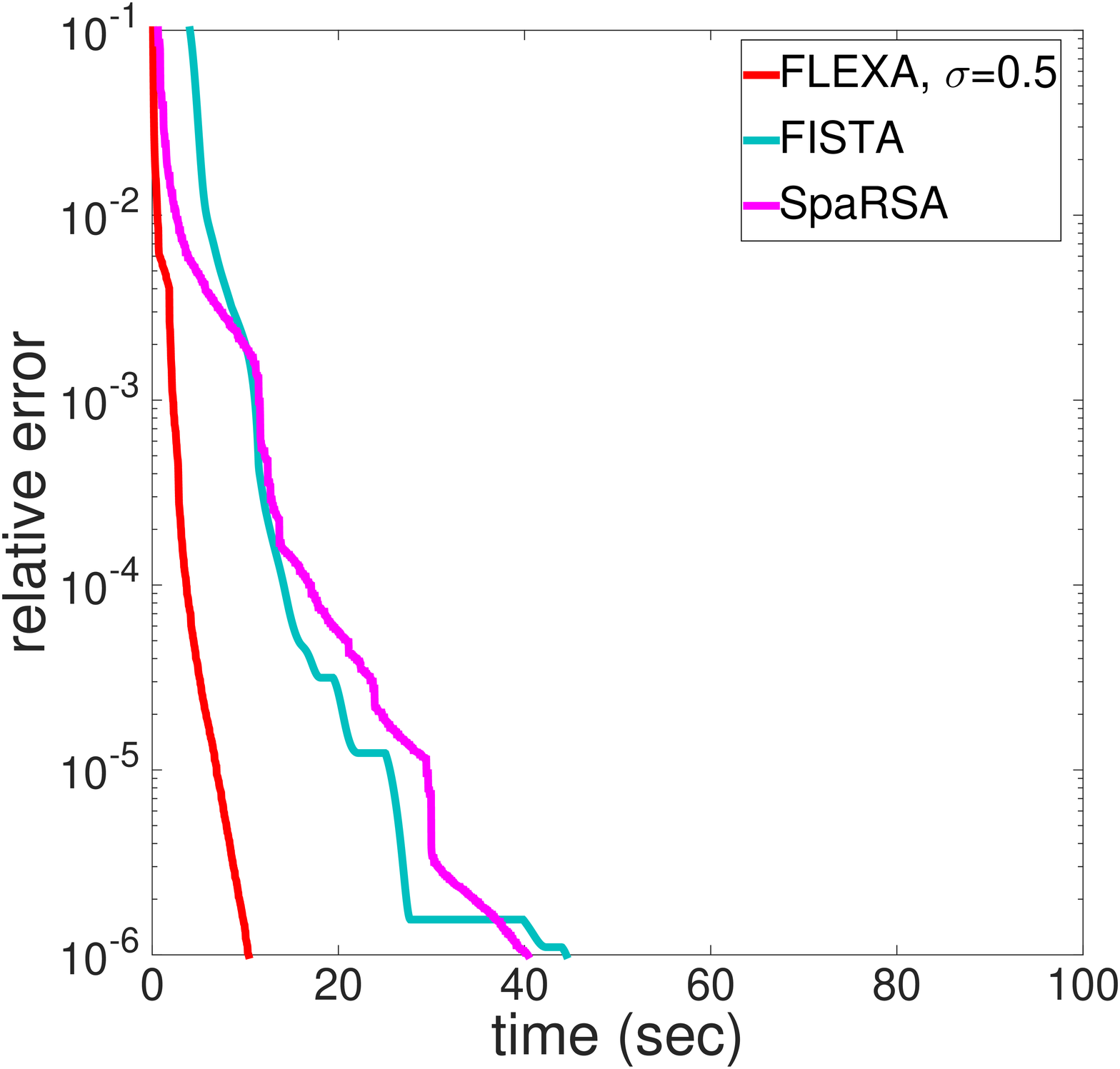}
      \end{subfigure}
      \begin{subfigure}[]{0.24\textwidth}
                \includegraphics[width=\textwidth]{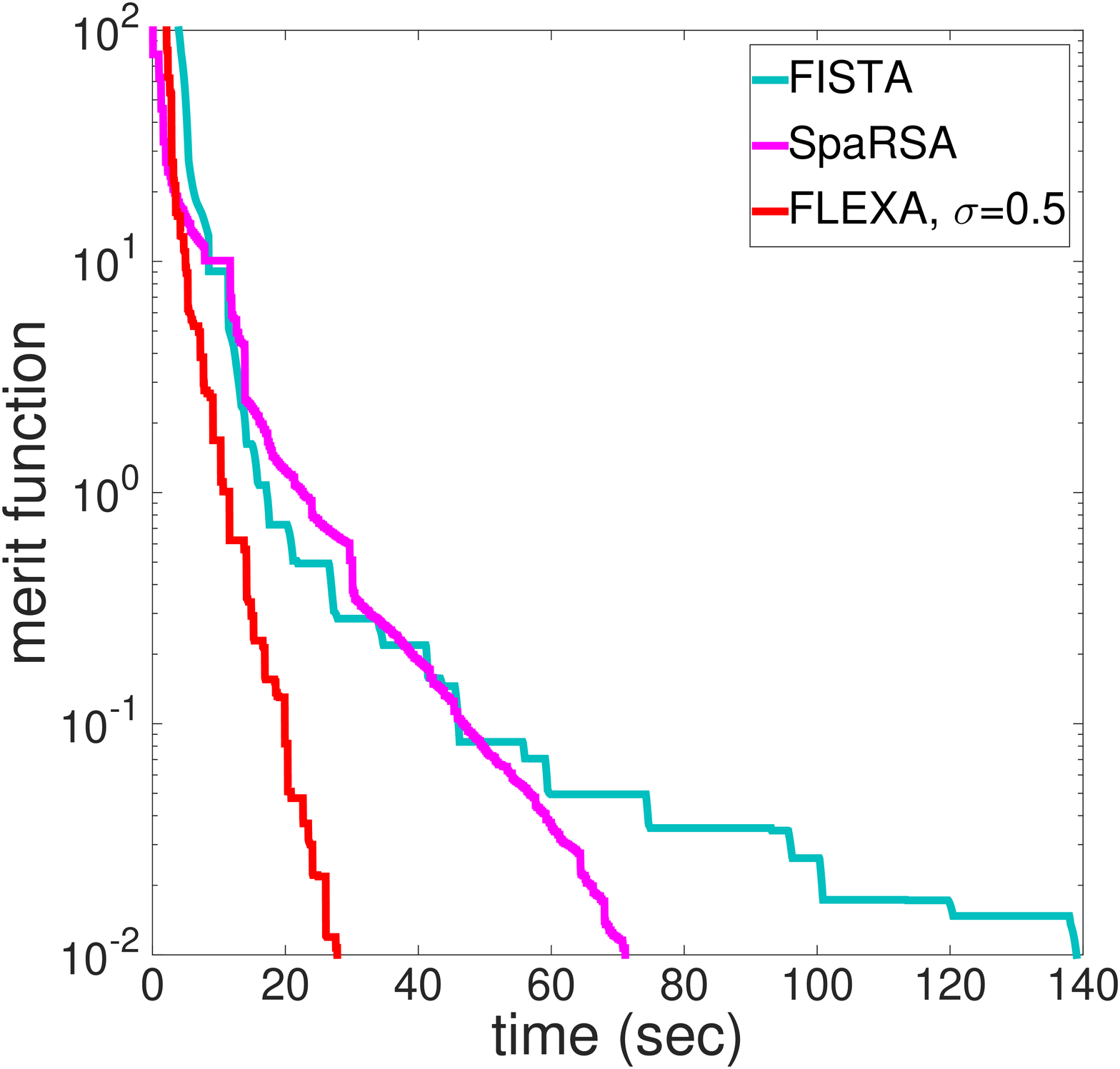}
      \end{subfigure}
\caption{Nonconvex problem (\ref{eq:problem_ncvx}) with 10\% solution sparsity:relative error vs. time (in seconds), and merit value vs. time (in seconds). \label{figN2}}
\end{figure}
These preliminary tests seem to indicate that FLEXA   performs well also on nonconvex problems.

\section{Conclusions}\label{sec:conclusions}
We  proposed a general  algorithmic framework  for the minimization of the sum of a possibly noncovex  differentiable function and a
possibily nonsmooth but  block-separable convex one.
Quite remarkably, our framework leads to different new algorithms whose degree of parallelism can be chosen by the user that also has 
complete control on the number of variables to update at each iterations; all the algorithms converge under the same conditions.
 Many well known  sequential and simultaneous solution methods in the literature are just  special cases of our algorithmic
 framework.
It is remarkable that our scheme encompasses many different practical realizations and that this flexibility can be
used to easily cope with different classes of problems that bring different challenges.
Our preliminary tests are very promising, showing that our schemes can outperform state-of-the-art algorithms.
Experiments on larger and more varied classes of problems (including those listed in  Sec. \ref{sec:Pro}) are the subject of
current research.  Among the topics that should be further studied, one key issue is the right initial choice and and subsequent 
tuning of the $\tau_i$. These parameters to a certaine extent determine the lenght of the shift at each iteration and 
play a crucial role in establishing the quality of the results.

\section{Acknowledgments}
The authors are very  grateful to Loris Cannelli and Paolo Scarponi for their invaluable help  in developing the C++ code of   the simulated 
algorithms.

The work of Facchinei was supported by the MIUR project PLATINO (Grant Agreement n. PON01\_01007). The work of Scutari was 
supported by the USA National Science Foundation under Grants CMS 1218717 and  CAREER Award No. 1254739.

\appendix
\section{Appendix: Proof of Theorem 1 and 2}
We first introduce some preliminary results instrumental to prove both Theorem 1 and Theorem 2.  Hereafter, for   notational simplicity,  we will omit the dependence of $\widehat{\mathbf{x}}(\mathbf y,\boldsymbol{\tau})$ on $\boldsymbol\tau$ and write $\widehat{\mathbf{x}}(\mathbf y)$.  Given $S\subseteq \mathcal N$ and $\x\triangleq (x_i)_{i=1}^N$, we will also denote by $(\mathbf x)_S$ (or interchangeably $\mathbf x_S$) the vector whose component  $i$ is equal to $x_i$ if $i\in S$, and zero otherwise.\vspace{-0.2cm}

\subsection{Intermediate results}

\begin{lemma}\label{Lemma_f_x_y_properties}  Let $H(\mathbf{x}; \mathbf{y}) \triangleq \sum_i h_i (\mathbf{x}_i; \mathbf{y})$. Then, the following hold:

\noindent (i) $H(\mathbf{\bullet};\mathbf{y})$
is uniformly strongly convex on $X$ with constant $c_{\boldsymbol{{\tau}}}>0$,
i.e., \vspace{-0.2cm}
\begin{equation}
\begin{array}{l}
\left(\mathbf{x}-\mathbf{w}\right)^{T}\left(\nabla_{\mathbf{x}}H\left(\mathbf{x};\mathbf{y}\right)-\nabla_{\mathbf{x}}H\left(\mathbf{w};\mathbf{y}\right)\right)\geq c_{{\boldsymbol{\tau}}}\left\Vert \mathbf{x}-\mathbf{w}\right\Vert ^{2}\end{array},\label{eq:strong_cvx_f_tilde}
\end{equation}
for all $\mathbf{x},\mathbf{w}\in X$ and given $\mathbf{y}\in X$;

\noindent  (ii) $\nabla_{\mathbf{x}}H(\mathbf{x};\mathbf{\bullet})$
is uniformly Lipschitz continuous on $X$, i.e., there exists
a  $0<L_{\nabla_H}<\infty$ independent on $\mathbf{x}$ such
that
\begin{equation}
\begin{array}{l}
\left\Vert \nabla_{\mathbf{x}}H\left(\mathbf{x};\mathbf{y}\right)-\nabla_{\mathbf{x}}H\left(\mathbf{x};\mathbf{w}\right)\right\Vert \end{array}\leq\, L_{\nabla H}\,\left\Vert \mathbf{y}-\mathbf{w}\right\Vert ,\label{eq:Lip_grad_L_f}
\end{equation}
for all $\mathbf{y},\mathbf{w}\in X$ and given $\mathbf{x}\in X$. \end{lemma}

\begin{proposition}\label{Prop_x_y} Consider  Problem (\ref{eq:problem 1})
under A1-A6.
Then the mapping  $X\ni\mathbf{y}\mapsto\widehat{\mathbf{x}}(\mathbf{y})$
has the following properties:

\noindent (a)\emph{ $\widehat{\mathbf{x}} (\mathbf{\bullet})$}
is Lipschitz continuous on\emph{ $X$, }i.e., there
exists a positive constant $\hat{{L}}$ such that\emph{
\begin{equation}
\left\Vert \widehat{\mathbf{x}}(\mathbf{y})-\widehat{\mathbf{x}}(\mathbf{z})\right\Vert \leq\,\hat{{L}}\,\left\Vert \mathbf{y}-\mathbf{z}\right\Vert ,\quad\forall\mathbf{y},\mathbf{z}\in X;\vspace{-0.3cm}\label{eq:Lipt_x_map}
\end{equation}
}

\noindent(b) the set of the fixed-points of\emph{ }$\widehat{\mathbf{x}}(\mathbf{\bullet})$\emph{
}coincides with the set of stationary solutions of Problem
(\ref{eq:problem 1}); therefore $\widehat{\mathbf{x}}(\bullet)$ has a fixed-point;

\noindent(c) for every given $\mathbf{y}\in X$   and for any set $S \subseteq {\cal N}$, 
\begin{align}
\left(\widehat{\mathbf{x}}(\mathbf{y})-\mathbf{y}\right)_S^{T}\,\nabla_{\mathbf{x}}F(\mathbf{y})_S +
& \sum_{i\in S} g_i(\widehat{\mathbf{x}}_i (\mathbf{y})) - \sum_{i\in S} g_i(\mathbf{y}_i)\label{eq:descent_direction}
\\ &\leq-c_\tau\,\left\Vert (\widehat{\mathbf{x}}(\mathbf{y})-\mathbf{y})_S\right\Vert ^{2}, \vspace{-0.3cm} \nonumber \vspace{-0.4cm}
\end{align}
with $c_\tau \triangleq q\, \min_i \tau_i$.
\end{proposition}
\begin{proof} We prove the proposition in the following order: (c), (a), (b).

\noindent (c): Given $\mathbf{y}\in X$, by definition,
each $\widehat{\mathbf{x}}_{i}(\mathbf{y})$ is the unique solution
of  problem (\ref{eq:decoupled_problem_i}); then it is not difficult to see that the following holds:
for all $\mathbf{z}_{i}\in X_{i}$,
\begin{equation}\label{eq:VI_i}
\left(  \mathbf{z}_i-\widehat{\mathbf{x}}_i(\mathbf{y})  \right)^T
\nabla_{\mathbf{x}_i}h_i(\widehat{\mathbf{x}}_{i}(\mathbf{y}); \mathbf{y}) + g_i(\mathbf{z}_i) - g_i(\widehat{\mathbf{x}}_{i}(\mathbf{y}))
\geq 0.
\end{equation}
Summing and subtracting $\nabla_{\mathbf{x}_{i}}P_i\left(\mathbf{y}_{i};\,\mathbf{y}\right)$
in (\ref{eq:VI_i}), choosing $\mathbf{z}_{i}=\mathbf{y}_{i}$, and
using  P2,
we get
\begin{equation}
\begin{array}[t]{l}
\left(\mathbf{y}_{i}-\widehat{\mathbf{x}}_{i}(\mathbf{y})\right)^{T}\left(\nabla_{\mathbf{x}_{i}}P_{i}\!\left(\widehat{\mathbf{x}}_{i}(\mathbf{y});
\,\mathbf{y}\right)-\nabla_{\mathbf{x}_{i}}P_i\!\left(\mathbf{y}_{i};\,\mathbf{y}\right)\right)\smallskip\\
\quad+\left(\mathbf{y}_{i}-\widehat{\mathbf{x}}_{i}(\mathbf{y})\right)^{T}\nabla_{\mathbf{x}_{i}}F(\mathbf{y})\smallskip
 + g_i(\mathbf{y}_i) - g_i(\widehat{\mathbf{x}}_{i}(\mathbf{y}))
\\
\quad-\tau_{i}\,(\widehat{\mathbf{x}}_{i}(\mathbf{y})-\mathbf{y}_{i})^{T}\,\mathbf{Q}_{i}(\mathbf{y})\,(\widehat{\mathbf{x}}_{i}(\mathbf{y})-\mathbf{y}_{i})\geq0,
\end{array}\label{eq:VI_i_row2}
\end{equation}
for all $i\in\mathcal{N}$. Observing that the term on the first line of
\eqref{eq:VI_i_row2} is non positive and using P1,
  we obtain
\begin{align*}
&\left(  \y_i-\widehat{\x}_i(\y)  \right)^T  \nabla_{\x_i}F(\y)
+ g_i(\y_i) - g_i(\widehat{\x}_i(\y))
\\ &\,
  \geq c_{\boldsymbol{\tau}}\left\Vert
 \widehat{\x}_i(\y)-\y_i\right
 \Vert ^{2},
\end{align*}
for all $i\in\mathcal{N}$. Summing  over $i\in S$
we get (\ref{eq:descent_direction}).

\noindent(a): We use the notation introduced in Lemma \ref{Lemma_f_x_y_properties}.
Given $\mathbf{y},\mathbf{z}\in X$, by optimality and  \eqref{eq:VI_i}, we have, for all
$\mathbf{v}$ and $\mathbf{w}$ in $ X$
\[ \begin{array}{l}
\left(\mathbf{v}-\widehat{\mathbf{x}}(\mathbf{y})\right)^{T}\nabla_{\mathbf{x}}H\left(\widehat{\mathbf{x}}(\mathbf{y});\mathbf{y}\right)
 +G(\mathbf{v}) - G(\widehat{\mathbf{x}}(\mathbf{y}))
  \geq  0
  \\
\left(\mathbf{w}-\widehat{\mathbf{x}}(\mathbf{z})\right)^{T}\nabla_{\mathbf{x}}H\left(\widehat{\mathbf{x}}(\mathbf{z});\mathbf{z}\right)
 + G(\mathbf{w}) - G(\widehat{\mathbf{x}}(\mathbf{z})) \geq  0.
\end{array}
\]
 Setting $\mathbf{v}=\widehat{\mathbf{x}}(\mathbf{z})$ and $\mathbf{w}=\widehat{\mathbf{x}}(\mathbf{y})$,
summing the two inequalities above, and adding and subtracting $\nabla_{\mathbf{x}}H\left(\widehat{\mathbf{x}}(\mathbf{y});\mathbf{z}\right)$,
we obtain:
\begin{equation}
\begin{array}{l}
\left(\widehat{\mathbf{x}}(\mathbf{z})-\widehat{\mathbf{x}}(\mathbf{y})\right)^{T}\left(\nabla_{\mathbf{x}}H\left(\widehat{\mathbf{x}}(\mathbf{z});
\mathbf{z}\right)-\nabla_{\mathbf{x}}H\left(\widehat{\mathbf{x}}(\mathbf{y});\mathbf{z}\right)\right)\\
\leq\left(\widehat{\mathbf{x}}(\mathbf{y})-\widehat{\mathbf{x}}(\mathbf{z})\right)^{T}\left(\nabla_{\mathbf{x}}H\left(\widehat{\mathbf{x}}(\mathbf{y});
\mathbf{z}\right)-\nabla_{\mathbf{x}}H\left(\widehat{\mathbf{x}}(\mathbf{y});\mathbf{y}\right)\right).
\end{array}\label{eq:minimum_principle_Lip}
\end{equation}
 Using (\ref{eq:strong_cvx_f_tilde}) we can  lower bound the left-hand-side
of (\ref{eq:minimum_principle_Lip}) as
\begin{equation}
\begin{array}{l}
\left(\widehat{\mathbf{x}}(\mathbf{z})-\widehat{\mathbf{x}}(\mathbf{y})\right)^{T}\left(\nabla_{\mathbf{x}}H\left(\widehat{\mathbf{x}}(\mathbf{z});\mathbf{z}\right)-
\nabla_{\mathbf{x}}H\left(\widehat{\mathbf{x}}(\mathbf{y});\mathbf{z}\right)\right)\\
\,\,\geq c_{{\boldsymbol{\tau}}}\left\Vert \widehat{\mathbf{x}}(\mathbf{z})-\widehat{\mathbf{x}}(\mathbf{y})\right\Vert ^{2},
\end{array}\label{eq:lipschtz_map_2}
\end{equation}
whereas the right-hand-side of (\ref{eq:minimum_principle_Lip}) can
be upper bounded as
\begin{equation}
\begin{array}{l}
\left(\widehat{\mathbf{x}}(\mathbf{y})-\widehat{\mathbf{x}}(\mathbf{z})\right)^{T}\left(\nabla_{\mathbf{x}}H\left(\widehat{\mathbf{x}}(\mathbf{y});\mathbf{z}\right)
-\nabla_{\mathbf{x}}H\left(\widehat{\mathbf{x}}(\mathbf{y});\mathbf{y}\right)\right)\\
\,\,\leq\, L_{\nabla H}\,\left\Vert \widehat{\mathbf{x}}(\mathbf{y})-\widehat{\mathbf{x}}(\mathbf{z})\right\Vert \,\left\Vert \mathbf{y}-\mathbf{z}\right\Vert ,
\end{array}\label{eq:lipschtz_map_1}
\end{equation}
where the inequality follows from the Cauchy-Schwartz inequality and
(\ref{eq:Lip_grad_L_f}). Combining (\ref{eq:minimum_principle_Lip}),
(\ref{eq:lipschtz_map_2}), and (\ref{eq:lipschtz_map_1}), we obtain
the desired Lipschitz property of $\widehat{\mathbf{x}}(\bullet)$.

\noindent (b): Let $\mathbf{x}^{\star}\in X$ be a fixed
point of $\widehat{\mathbf{x}}(\mathbf{y})$, that is $\mathbf{x}^{\star}=\widehat{\mathbf{x}}(\mathbf{x}^{\star})$.
Each $\widehat{\mathbf{x}}_{i}(\mathbf{y})$ satisfies (\ref{eq:VI_i}) for any given $\mathbf{y}\in X$.
For some $\boldsymbol\xi_i \in \partial g_i(\mathbf{x}^*)$, setting
$\mathbf{y}=\mathbf{x}^{\star}$ and using $\mathbf{x}^{\star}=\widehat{\mathbf{x}}(\mathbf{x}^{\star})$ and the convexity of $g_i$,
(\ref{eq:VI_i}) reduces to
\begin{equation}
\left(\mathbf{z}_{i}-\mathbf{x}_{i}^{\star}\right)^{T}(\nabla_{\mathbf{x}_{i}}F(\mathbf{x}^{\star}) + \boldsymbol\xi_i) \geq 0,\label{eq:fixed_point_min_principle}
\end{equation}
for all $\mathbf{z}_{i}\in X_{i}$ and $i\in\mathcal{N}$.
Taking into account the Cartesian structure of $X$, the separability of $G$,  and
summing (\ref{eq:fixed_point_min_principle}) over $i\in\mathcal{N}$,
we obtain $\begin{array}[t]{l}
\left(\mathbf{z}-\mathbf{x}^{\star}\right)^{T}(\nabla_{\mathbf{x}}F(\mathbf{x}^{\star}) + \boldsymbol\xi) \geq0,\end{array}$ for all $\mathbf{z}\in X,$ with $\mathbf{z}\triangleq(\mathbf{z}_{i})_{i=1}^{N}$ and $\boldsymbol\xi \triangleq(\boldsymbol\xi_i)_{i=1}^{N} \in \partial G(\mathbf{x}^*)$;
therefore $\mathbf{x}^{\star}$ is a stationary solution of (\ref{eq:problem 1}).

The converse holds because i) $\widehat{\mathbf{x}}(\mathbf{x}^{\star})$
is the unique optimal solution of (\ref{eq:decoupled_problem_i})
with $\mathbf{y}=\mathbf{x}^{\star}$, and ii) $\mathbf{x}^{\star}$
is also an optimal solution of (\ref{eq:decoupled_problem_i}), since
it satisfies the minimum principle.\end{proof}

\begin{lemma} \emph{\cite[Lemma 3.4, p.121]{Bertsekas-Tsitsiklis_bookNeuro11}}\label{lemma_Robbinson_Siegmunt}
Let $\{X^{k}\}$, $\{Y^{k}\}$, and $\{Z^{k}\}$ be three sequences
of numbers such that $Y^{k}\geq0$ for all $k$. Suppose that
\[
X^{k+1}\leq X^{k}-Y^{k}+Z^{k},\quad\forall k=0,1,\ldots
\]
and $\sum_{k=0}^\infty Z^{k}<\infty$. Then either $X^{k}\rightarrow-\infty$
or else $\{X^{k}\}$ converges to a finite value and $\sum_{k=0}^\infty Y^{k}<\infty$.
\end{lemma}

\begin{lemma}\label{lemma on errors} Let $\{\x^k\}$ be the sequence generated by Algorithm 1. Then, there is a positive constant $\tilde c$ such that the following holds: for all $k\geq 1$,
\begin{equation}\label{error_descent}
\begin{array}{ll}
\left( \nabla_{\x}F (\x^{k})\right)^{T}_{\tiny {S^k}} \left(\widehat{\x}(\x^k) -\x^k \right)_{\tiny {S^k}}
+\displaystyle{\sum_{i \in S^k}} g_i(\widehat{\x}_i(\x^k)) \medskip\\\quad\quad\quad\quad\quad -\displaystyle{\sum_{i \in S^k}} g_i(\x_i^k)
 \leq
-\tilde c \,\| \widehat{\x}(\x^k)-\x^k \|^2.
\end{array}
\end{equation}
\end{lemma}
\begin{proof} Let $j_k$ be an index in $S^k$ such that $E_{j_k}(\x^k) \geq \rho \max_i E_i(\x^k)$ (Step 2 of Algorithm 1). Then, using the aforementioned bound and (\ref{eq:error bound}),
it is easy to check that the following chain of inequalities holds:
\begin{align*}
\bar s_{j_k} \| \widehat{\x}_{S^k}(\x^k)  - \x^k_{S^k}\| & \geq \bar s_{j_k} \| \widehat{\x}_{j_k}(\x^k)  - \x^k_{j_k}\| \\ & \geq E_{j_k}(\x^k) \\ &
\geq \rho \max_i E_i(\x^k) \\ & \geq \left( \rho \min_i \underbar s_i\right) \left( \max_i \{ \|\widehat{\x}_i(\x^k)  - \x^k_i\| \} \right) \\ &
\geq \left( \frac{\rho \min_i \underbar s_i}{N}\right) \|\widehat{\x}(\x^k)  - \x^k\|.
\end{align*}
Hence we have for any $k$,
\begin{equation}\label{lower_bound_error_S_k}
\| \widehat{\x}_{S^k}(\x^k)  - \x^k_{S^k}\| \geq \left(\frac{\rho \min_i \underbar s_i}{N \bar s_{j_k}}\right) \|\widehat{\x}(\x^k)  - \x^k\|.
\end{equation}
Invoking now Proposition \ref{Prop_x_y}(c) with $S=S^k$ and $\y=\x^k$, and using \eqref{lower_bound_error_S_k}, (\ref{error_descent}) holds true, with $\tilde c \triangleq c_\tau \left(\frac{\rho \min_i \underbar s_i}{N \max_j \bar s_j}\right)^2$. \end{proof}

\subsection{Proof of Theorem \ref{Theorem_convergence_inexact_Jacobi}}\label{proof_Th1} \vspace{-0.2cm}We are now ready to prove the theorem.
For any given $k\geq 0$, the Descent Lemma \cite{Bertsekas_Book-Parallel-Comp}
yields
{\color{black}
\begin{equation}
\begin{array}{lll}
F\left(\mathbf{x}^{k+1}\right) & \leq & F\left(\mathbf{x}^{k}\right)+\gamma^{k}\,\nabla_{\mathbf{x}}F\left(\mathbf{x}^{k}\right)^{T}\left(\widehat{\mathbf{z}}^{k}-\mathbf{x}^{k}\right)\smallskip\\
 &  & +\dfrac{\left(\gamma^{k}\right)^{2}{L_{\nabla F}}}{2}\,\left\Vert \widehat{\mathbf{z}}^{k}-\mathbf{x}^{k}\right\Vert ^{2},
\end{array}\label{eq:descent_Lemma} 
\end{equation}
with $\widehat{\mathbf{z}}^{k}\triangleq(\widehat{\mathbf{z}}_{i}^{k})_{i=1}^{N}$ and
$\mathbf{z}^{k}\triangleq(\mathbf{z}_{i}^{k})_{i=1}^{N}$ defined in Step 2 and 3  (Algorithm \ref{alg:general}).
Observe that
\begin{equation}\label{eq:49bis} \begin{array}{rcl}
 \left\Vert \widehat{\mathbf{z}}^{k}-\mathbf{x}^{k}\right\Vert ^{2}&\leq&
 \left\Vert \mathbf{z}^{k}-\mathbf{x}^{k}\right\Vert ^{2}\\
 & \leq &
 2\left\Vert \widehat{\mathbf{x}}(\mathbf{x}^{k})-\mathbf{x}^{k}\right\Vert ^{2}
\\ &&
 +2\sum_{i\in \cal N}\left\Vert \mathbf{z}_{i}^{k}-\widehat{\mathbf{x}}_{i}(\mathbf{x}^{k})\right\Vert ^{2}\\
& \leq& 2\left\Vert \widehat{\mathbf{x}}(\mathbf{x}^{k})-\mathbf{x}^{k}\right\Vert ^{2}+2\sum_{i \in \cal N}(\varepsilon_{i}^{k})^{2},
 \end{array}
 \end{equation}
 where the first inequality follows from the definition of $ \mathbf{z}^{k}$ and $ \widehat{\mathbf{z}}^{k}$, and
 in the last inequality we used $\left\Vert \mathbf{z}_{i}^{k}-\widehat{\mathbf{x}}_{i}(\mathbf{x}^{k})\right\Vert \leq\varepsilon_{i}^{k}$.

Denoting by $\overline{S}^k$ the complement of $S$, we also have, for all $k$, 
\begin{equation}\label{nabla_times_z_hat}
\begin{array}{l}
\nabla_{\mathbf{x}}F\left(\mathbf{x}^{k}\right)^{T}\left(\widehat{\mathbf{z}}^{k}-\mathbf{x}^{k}\right)\smallskip\\
\qquad\qquad   =
\nabla_{\mathbf{x}}F\left(\mathbf{x}^{k}\right)^{T}\left(\widehat{\mathbf{z}}^{k}-\widehat{\mathbf{x}}(\mathbf{x}^{k})
+\widehat{\mathbf{x}}(\mathbf{x}^{k})-\mathbf{x}^{k}\right)\smallskip\\
\qquad\qquad  = \nabla_{\mathbf{x}}F\left(\mathbf{x}^{k}\right)^{T}_{S^k} (\mathbf{z}^k - \widehat{\mathbf{x}}(\mathbf{x}^k))_{S^k}\smallskip\\ \qquad\qquad\quad
 +
 \nabla_{\mathbf{x}}F\left(\mathbf{x}^{k}\right)^{T}_{\overline{S}^k} (\mathbf{x}^k - \widehat{\mathbf{x}}(\mathbf{x}^k))_{\overline{S}^k}\smallskip\\
\qquad \qquad \quad   +  \nabla_{\mathbf{x}}F\left(\mathbf{x}^{k}\right)^{T}_{S^k} (\widehat{\mathbf{x}}(\mathbf{x}^k)-\mathbf{x}^k)_{S^k}\smallskip\\ \qquad\qquad\quad
 +    \nabla_{\mathbf{x}}F\left(\mathbf{x}^{k}\right)^{T}_{\overline{S}^k} (\widehat{\mathbf{x}}(\mathbf{x}^k)-\mathbf{x}^k)_{\overline{S}^k}\smallskip\\
\qquad\qquad   = \nabla_{\mathbf{x}}F\left(\mathbf{x}^{k}\right)^{T}_{S^k} (\mathbf{z}^k - \widehat{\mathbf{x}}(\mathbf{x}^k))_{S^k} \smallskip\\ \qquad\qquad +
 \nabla_{\mathbf{x}}F\left(\mathbf{x}^{k}\right)^{T}_{S^k} (\widehat{\mathbf{x}}(\mathbf{x}^k)-\mathbf{x}^k)_{S^k},
\end{array}
\end{equation}
where in the second equality we used the definition of $\widehat{\mathbf{z}}^k$ and of the set $S^k$.
Now, using  (\ref{nabla_times_z_hat}) and
Lemma \ref{lemma on errors}, we can write
\begin{equation}
\begin{array}{rl}
\multicolumn{2}{l}{
\nabla_{\mathbf{x}}F\left(\mathbf{x}^{k}\right)^{T}\left(\widehat{\mathbf{z}}^{k}-\mathbf{x}^{k}\right)  + \sum_{i\in S^k} g_i(\widehat{\mathbf{z}}_i^{k})-\sum_{i\in S^k} g_i(\mathbf{x}^{k}_i) } \\[0.3em]
=& \nabla_{\mathbf{x}}F\left(\mathbf{x}^{k}\right)^{T}\left(\widehat{\mathbf{z}}^{k}-\mathbf{x}^{k}\right)  +
\sum_{i\in S^k} g_i(\widehat{\mathbf{x}}_i(\mathbf{x}^{k})) \\[0.3em]
&-\sum_{i\in S^k} g_i(\mathbf{x}^{k}_i)   +
 \sum_{i \in S^k} g_i(\widehat{\mathbf{z}}_i^k) -\sum_{i \in S^k} g_i(\widehat{\mathbf{x}}_i(\mathbf{x}^{k}))
 \end{array}
 \end{equation}
 \begin{equation}\label{eq:descent_at_x_n}
\begin{array}{rl}
\leq& -\tilde c\left\Vert \widehat{\mathbf{x}}(\mathbf{x}^{k})-\mathbf{x}^{k}\right
\Vert ^{2}+\sum_{i \in S^k}\varepsilon_{i}^{k}\left\Vert \nabla_{\mathbf{x}_{i}}F(\mathbf{x}^{k})\right\Vert
\\[0.3em]
&+ L_G\sum_{i \in S^k}\varepsilon_{i}^{k},
\end{array}
\end{equation}
where $L_G$ is a (global) Lipschitz constant for (all) $g_i$.\\
\indent Finally,  from    the definition of
 $\widehat{\z}^k$ and of the set $S^k$,  we have for all $k$, 
\begin{equation}
\begin{array}{l}
V(\mathbf{x}^{k+1})  =  F(\mathbf{x}^{k+1}) + \sum_{ {i \in \cal N}}g_i (\mathbf{x}_i^{k+1}) \\[0.3em]
 =    F(\mathbf{x}^{k+1}) + \sum_{i\in {\cal N}}g_i (\mathbf{x}_i^{k} + \gamma^k(\widehat{\mathbf{z}}^k_i - \mathbf{x}_i^{k} ))\\[0.3em]
\leq   F(\mathbf{x}^{k+1}) +  \sum_{i \in \cal N}g_i (\mathbf{x}^k_i) + \gamma^k \left(\sum_{i \in S^k} (g_i (\widehat{\mathbf{z}}^k_i)
- g_i (\mathbf{x}^k_i)) \right)\\[0.3em]
\leq
V\left(\mathbf{x}^{k}\right)-\gamma^{k}\left(\tilde c -\gamma^{k}{L_{\nabla F}}\right)\left\Vert \widehat{\mathbf{x}}
(\mathbf{x}^{k})-\mathbf{x}^{k}\right\Vert ^{2}+T^{k},\end{array}\label{eq:descent_Lemma_2}
\end{equation}
where in the first inequality we used the the convexity of the $g_i$'s, whereas the second one follows from
 \eqref{eq:descent_Lemma}, \eqref{eq:49bis} and \eqref{eq:descent_at_x_n},  with
$$T^{k}\triangleq\gamma^{k}\,
\sum_{i \in S^k}\varepsilon_{i}^{k}\left( L_G +
\left\Vert \nabla_{\mathbf{x}_{i}}F(\mathbf{x}^{k})\right\Vert\right) +\left(\gamma^{k}\right)^{2}{L_{\nabla F}}\,\sum_{i\in \cal N}
(\varepsilon_{i}^{k})^{2}.$$\\
\indent Using  assumption  (iv),  
we can bound $T^k$ as
\[
T^{k}\leq
(\gamma^k)^2 \left[ N \alpha_1 (\alpha_2  L_G +1) + (\gamma^k)^2 L_{\nabla F} \left(N\alpha_1\alpha_2\right)^2 \right],
\]
 which, by assumption (iii) implies
 $\sum_{k=0}^{\infty}T^{k}<\infty$.
Since $\gamma^{k}\rightarrow 0$, it follows from (\ref{eq:descent_Lemma_2}) that there exist  some positive constant
$\beta_{1}$ and a sufficiently large $k$, say $\bar{{k}}$, such that
\begin{equation}
V(\mathbf{x}^{k+1})\leq V(\mathbf{x}^{k})-\gamma^{k}\beta_{1}\left\Vert \widehat{\mathbf{x}}(\mathbf{x}^{k})-\mathbf{x}^{k}\right\Vert ^{2}+T^{k},\label{eq:descent_Lemma_3_}
\end{equation}
for all $k\geq\bar{{k}}$. Invoking Lemma \ref{lemma_Robbinson_Siegmunt} with the identifications
$X^{k}=V(\mathbf{x}^{k+1})$, $Y^{k}=\gamma^{k}\beta_{1}\left\Vert \widehat{\mathbf{x}}(\mathbf{x}^{k})-\mathbf{x}^{k}\right\Vert ^{2}$
and $Z^{k}=T^{k}$ while using $\sum_{k=0}^\infty T^{k}<\infty$, we deduce
from (\ref{eq:descent_Lemma_3_}) that either $\{V(\mathbf{x}^{k})\}\rightarrow-\infty$
or else $\{V(\mathbf{x}^{k})\}$ converges to a finite
value and 
\begin{equation}
\lim_{k\rightarrow\infty}\sum_{t=\bar{{k}}}^{k}\gamma^{t}\left\Vert \widehat{\mathbf{x}}(\mathbf{x}^{t})-\mathbf{x}^{t}\right\Vert ^{2}<+\infty.\vspace{-0.1cm}\label{eq:finite_sum_series}
\end{equation}
Since $V$ is coercive, $V(\mathbf{x})\geq\min_{\mathbf{y}\in X}V(\mathbf{y})>-\infty$,
implying that $\{V\left(\mathbf{x}^{k}\right)\}$ is convergent;
it follows from (\ref{eq:finite_sum_series}) and $\sum_{k=0}^{\infty}\gamma^{k}=\infty$
that $\liminf_{k\rightarrow\infty}\left\Vert \widehat{\mathbf{x}}(\mathbf{x}^{k})-\mathbf{x}^{k}\right\Vert =0.$\\
\indent Using Proposition \ref{Prop_x_y}, we show next that $\lim_{k\rightarrow\infty}\left\Vert \widehat{\mathbf{x}}(\mathbf{x}^{k})-\mathbf{x}^{k}\right\Vert =0$;
for notational simplicity we will write $\triangle\widehat{\mathbf{x}}(\mathbf{x}^{k})\triangleq\widehat{\mathbf{x}}(\mathbf{x}^{k})-\mathbf{x}^{k}$.
Suppose, by contradiction, that $\limsup_{k\rightarrow\infty}\left\Vert \triangle\widehat{\mathbf{x}}(\mathbf{x}^{k})\right\Vert >0$.
Then, there exists a $\delta>0$ such that $\left\Vert \triangle\widehat{\mathbf{x}}(\mathbf{x}^{k})\right\Vert >2\delta$
for infinitely many $k$ and also $\left\Vert \triangle\widehat{\mathbf{x}}(\mathbf{x}^{k})\right\Vert <\delta$
for infinitely many $k$. Therefore, one can always find an infinite
set of indexes, say $\mathcal{K}$, having the following properties:
for any $k\in\mathcal{K}$, there exists an integer $i_{k}>k$ such
that
\begin{eqnarray}
\left\Vert \triangle\widehat{\mathbf{x}}(\mathbf{x}^{k})\right\Vert <\delta, &  & \left\Vert \triangle\widehat{\mathbf{x}}(\mathbf{x}^{i_{k}})\right\Vert >2\delta\medskip\label{eq:outside_interval}\\
\delta\leq\left\Vert \triangle\widehat{\mathbf{x}}(\mathbf{x}^{j})\right\Vert \leq2\delta &  & k<j<i_{k}.\label{eq:inside_interval}
\end{eqnarray}
Given the above bounds, the following holds: for all $k\in\mathcal{K}$,
\begin{eqnarray}
\delta & \overset{(a)}{<} & \left\Vert \triangle\widehat{\mathbf{x}}(\mathbf{x}^{i_{k}})\right\Vert -\left\Vert \triangle\widehat{\mathbf{x}}(\mathbf{x}^{k})\right\Vert \medskip\nonumber \\
 & \leq & \left\Vert \widehat{\mathbf{x}}(\mathbf{x}^{i_{k}})-\widehat{\mathbf{x}}(\mathbf{x}^{k})\right\Vert +\left\Vert \mathbf{x}^{i_{k}}-\mathbf{x}^{k}\right\Vert \nonumber\\
 & \overset{(b)}{\leq} & (1+\hat{{L}})\left\Vert \mathbf{x}^{i_{k}}-\mathbf{x}^{k}\right\Vert \nonumber\end{eqnarray}
 \begin{eqnarray}
 & \overset{(c)}{\leq} &    {\color{black} (1+\hat{{L}})\sum_{t=k}^{i_{k}-1}\gamma^{t}\left(\left\Vert \triangle\widehat{\mathbf{x}}(\mathbf{x}^{t})_{S^t}\right\Vert +\left\Vert (\mathbf{z}^{t}-\widehat{\mathbf{x}}(\mathbf{x}^{t}))_{S^t}\right\Vert \right)}\vspace{-0.3cm}\nonumber \\
 & \overset{(d)}{\leq} & (1+\hat{{L}})\,(2\delta+\varepsilon^{\max})\sum_{t=k}^{i_{k}-1}\gamma^{t},\label{eq:lower_bound_sum}
\end{eqnarray}
where (a) follows from (\ref{eq:outside_interval});
(b) is due to Proposition \ref{Prop_x_y}(a); (c) comes from the triangle
inequality, the updating rule of the algorithm {\color{black} and the definition of $\widehat{\z}^k$}; and in (d) we used
(\ref{eq:outside_interval}), (\ref{eq:inside_interval}), and $\left\Vert \mathbf{z}^{t}-\widehat{\mathbf{x}}(\mathbf{x}^{t})\right\Vert \leq\sum_{i\in \cal N}\varepsilon_{i}^{t}$,
where $\varepsilon^{\max}\triangleq\max_{k}\sum_{i\in\cal N}\varepsilon_{i}^{k}<\infty$.
It follows from (\ref{eq:lower_bound_sum}) that
\begin{equation}
\liminf_{k\rightarrow\infty}\sum_{t=k}^{i_{k}-1}\gamma^{t}\geq\dfrac{{\delta}}{(1+\hat{{L}})(2\delta+\varepsilon^{\max})}>0.\label{eq:lim_inf_bound}
\end{equation}

\noindent
We show next that (\ref{eq:lim_inf_bound}) is in contradiction with
the convergence of $\{V(\mathbf{x}^{k})\}$. To do that, we preliminary
prove that, for sufficiently large $k\in\mathcal{K}$, it must be
$\left\Vert \triangle\widehat{\mathbf{x}}(\mathbf{x}^{k})\right\Vert \geq\delta/2$.
Proceeding as in (\ref{eq:lower_bound_sum}), we have: for any given
$k\in\mathcal{K}$,
\[
\begin{array}{l}
\left\Vert \triangle\widehat{\mathbf{x}}(\mathbf{x}^{k+1})\right\Vert -\left\Vert \triangle\widehat{\mathbf{x}}(\mathbf{x}^{k})\right\Vert \leq(1+\hat{{L}})\left\Vert \mathbf{x}^{k+1}-\mathbf{x}^{k}\right\Vert \smallskip\\
\qquad\qquad\qquad\qquad\qquad\leq(1+\hat{{L}})\gamma^{k}\left(\left\Vert \triangle\widehat{\mathbf{x}}(\mathbf{x}^{k})\right\Vert +\varepsilon^{\max}\right).
\end{array}
\]
 It turns out that for sufficiently large $k\in\mathcal{K}$ so that
$(1+\hat{{L}})\gamma^{k}<\delta/(\delta+2\varepsilon^{\max})$, it
must be
\begin{equation}
\left\Vert \triangle\widehat{\mathbf{x}}(\mathbf{x}^{k})\right\Vert \geq\delta/2;\label{eq:lower_bound_delta_x_n}
\end{equation}
otherwise the condition $\left\Vert \triangle\widehat{\mathbf{x}}(\mathbf{x}^{k+1})\right\Vert \geq\delta$
would be violated {[}cf. (\ref{eq:inside_interval}){]}. Hereafter
we assume without loss of generality  that (\ref{eq:lower_bound_delta_x_n}) holds for
all $k\in\mathcal{K}$ (in fact, one can alway restrict $\{\mathbf{x}^{k}\}_{k\in\mathcal{K}}$
to a proper subsequence).

We can show now that (\ref{eq:lim_inf_bound}) is in contradiction
with the convergence of $\{V(\mathbf{x}^{k})\}$. Using (\ref{eq:descent_Lemma_3_})
(possibly over a subsequence), we have: for sufficiently large $k\in\mathcal{K}$,
\begin{eqnarray}
V(\mathbf{x}^{i_{k}}) & \leq & V(\mathbf{x}^{k})-\beta_{2}\sum_{t=k}^{i_{k}-1}\gamma^{t}\left\Vert \triangle\widehat{\mathbf{x}}(\mathbf{x}^{t})\right\Vert ^{2}+\sum_{t=k}^{i_{k}-1}T^{t}\nonumber \\
 & \overset{(a)}{<} & V(\mathbf{x}^{k})-\beta_{2}(\delta^{2}/4)\sum_{t=k}^{i_{k}-1}\gamma^{t}+\sum_{t=k}^{i_{k}-1}T^{t}\label{eq:liminf_zero}
\end{eqnarray}
where in (a) we used (\ref{eq:inside_interval}) and (\ref{eq:lower_bound_delta_x_n}),
and $\beta_{2}$ is some positive constant. Since $\{V(\mathbf{x}^{k})\}$
converges and $\sum_{k=0}^{\infty}T^{k}<\infty$, (\ref{eq:liminf_zero})
implies $\lim_{\mathcal{K}\ni k\rightarrow\infty}\,\sum_{t=k}^{i_{k}-1}\gamma^{t}=0,$
which contradicts (\ref{eq:lim_inf_bound}).
Finally, since the sequence $\{\mathbf{x}^{k}\}$ is bounded {[}by the coercivity of $V$ and the convergence of $\{V(\mathbf{x}^{k})\}${]},
it has at least one limit point $\bar{{\mathbf{x}}}$, that  belongs
to $X$. By the continuity of $\widehat{\mathbf{x}}(\bullet)$
{[}Proposition \ref{Prop_x_y}(a){]} and $\lim_{k\rightarrow\infty}\left\Vert \widehat{\mathbf{x}}(\mathbf{x}^{k})-\mathbf{x}^{k}\right\Vert =0$,
it must be $\widehat{\mathbf{x}}(\bar{{\mathbf{x}}})=\bar{{\mathbf{x}}}$.
By Proposition \ref{Prop_x_y}(b) $\bar{{\mathbf{x}}}$ is  a stationary
solution of Problem  (\ref{eq:problem 1}).
As a final remark, note that if $\varepsilon_i^k=0$ for every $i$ and for every $k$ large enough, i.e., if eventually
$\widehat{\mathbf{x}}(\mathbf{x}^k)$ is computed exactly, there is no need to assume that $G$ is
globally Lipschitz. In fact  in \eqref{eq:descent_at_x_n} the term containing $L_G$ disappears,
all  $T^k$ are zero and all the subsequent derivations hold independent of the Lipschitzianity of $G$.
\hfill$\square$

\subsection{Proof of Theorem \ref{Theorem_convergence_Gauss_Jacobi}}\label{proof_Th2}
We show next that  Algorithm \ref{alg:PJGA}
is just an  instance of the inexact Jacobi scheme described in Algorithm
\ref{alg:general} satisfying  the convergence conditions in Theorem
\ref{Theorem_convergence_inexact_Jacobi}; which proves Theorem \ref{Theorem_convergence_Gauss_Jacobi}.  It is not difficult to see that this boils down  to proving
that, for all $p\in P$ and $i\in I_p$, the sequence $\mathbf{z}_{pi}^{k}$
in Step 2a) of Algorithm \ref{alg:PJGA} satisfies
\begin{equation}
\|\mathbf{z}_{pi}^k-\widehat{{\x}}_{pi}(\mathbf{x}^{k})\|\leq\tilde{\varepsilon}_{pi}^{k},
\label{eq:GS_Jacobi_error_bound}
\end{equation}
for some $\{\tilde{\varepsilon}_{pi}^{k}\}$ such that $\sum_{n}\tilde{\varepsilon}_{pi}^{k}\,
\gamma^{k}<\infty$.
The following holds for the LHS of (\ref{eq:GS_Jacobi_error_bound}):\vspace{-0.2cm}
\begin{equation}\label{eq:error_GS}
\begin{array}{l}
\!\!\!\!\|\mathbf{z}_{pi}^{k}-\widehat{\mathbf{x}}_{pi}(\mathbf{x}^{k})\|\leq\|\widehat{\mathbf{x}}_{pi}
(\mathbf{x}_{pi<}^{k+1},\mathbf{x}_{pi\geq}^{k}, \mathbf{x}_{-p}^k)-
\widehat{\mathbf{x}}_{pi}(\mathbf{x}^{k})\|
\\[0.5em]
\quad\quad\quad\quad\quad\quad\,\quad +\, \|
\mathbf{z}_{pi}^{k}-\widehat{\mathbf{x}}_{pi}(\mathbf{x}_{pi<}^{k+1},\mathbf{x}_{pi\geq}^{k}, \mathbf{x}_{-p}^k)\|\smallskip\\[0.5em]
\quad\overset{(a)}{\leq}\|\widehat{\mathbf{x}}_{pi}(\mathbf{x}_{pi<}^{k+1},\mathbf{x}_{pi\geq}^{k}, \mathbf{x}_{-p}^k)-
\widehat{\mathbf{x}}_{pi}(\mathbf{x}^{k})\|+{\varepsilon}_{pi}^{k}
\\[0.5em]\quad\overset{(b)}{\le}\hat{{L}}\,\|\mathbf{x}_{pi<}^{k+1}-\mathbf{x}_{pi<}^{k}\|+{\varepsilon}_{pi}^{
k}\smallskip\\[0.5em]
\quad\overset{(c)}{=}
\hat{{L}}\gamma^{k} \left\Vert
\left(\mathbf{z}_{pi<}^{k}-\mathbf{x}_{pi<}^k\right)\right\Vert
+{\varepsilon}_{pi}^{k}\\[0.5em]
\quad\leq\hat{{L}}\gamma^{k}\left(
\sum_{j=1}^{i-1}( \|\mathbf{z}^k_{pj}- \widehat{\mathbf{x}}_{pj}(\mathbf{x}^k)\| + \|\widehat{\mathbf{x}}_{pj}(\mathbf{x}^k) - \mathbf{x}^k_{pj}\|) \right)
\\[0.5em]
\quad\quad +\,{\varepsilon}_{pi}^{k}\\[0.5em]
\quad\overset{(d)}{\leq}\hat{{L}}\gamma^{k}\sum_{j=1}^{i-1}
{\varepsilon}_{pj}^{k}+\hat{{L}}\gamma^{k}\left(\dfrac{L_G+\beta}{c_{\boldsymbol{\tau}}}\right)+{\varepsilon}_{pi}^{k}\triangleq \tilde{\varepsilon}_{pi}^{k},
\end{array}
\end{equation}
where (a) follows from the error bound in Step 2a) of Algorithm \ref{alg:PJGA};
in (b) we used Proposition \ref{Prop_x_y}a); (c) is due to Step 2b);
and  (d) follows from $\|\widehat{\mathbf{x}}_{pj}(\mathbf{x}^k) - \mathbf{x}^k_{pj}\|\leq (L_G + \|\nabla_{\x_{pi}} F(\x^k)\|)/c_{\boldsymbol{\tau}}$ and $\|\nabla_{\x_{pi}} F(\x^k)\|\leq \beta$ for some $0< \beta <\infty$.  It turns out that (\ref{eq:GS_Jacobi_error_bound})
is satisfied choosing $\tilde{\varepsilon}_{pi}^{k}$ as defined in the last inequality of \eqref{eq:error_GS}. 

\bibliographystyle{IEEEtran}
\balance
\bibliography{scutari_facchinei_refs}

\end{document}